%% file: root.tex
\documentclass[journal]{ieeetran}  % Comment this line out
\IEEEoverridecommandlockouts                              % This command is only
\usepackage{amsmath} % assumes amsmath package installed
\usepackage{amssymb}  % assumes amsmath package installed
\usepackage{graphicx}
\usepackage[usenames,dvipsnames,svgnames,table]{xcolor}
\usepackage{enumitem}

\usepackage{color}

\usepackage{subcaption}
\usepackage{comment}

\newcommand{\eps}{\epsilon}

\newcommand{\E}{\mathbb{E}}

\newcommand{\R}{\mathbb{R}}

\newcommand{\Bc}{\mathcal{B}}

\newcommand{\Dc}{\mathcal{D}}

\newcommand{\Sc}{\mathcal{S}}

\newcommand{\bs}{\boldsymbol}
\newcommand{\mc}{\mathcal}
\newcommand{\mb}{\mathbb}

\newcommand{\diag}{\mathrm{diag}}

\newcommand{\dpara}{a}

\newcommand{\comp}{\zeta}
\newcommand{\src}{s}
\newcommand{\buy}{b}
\newcommand{\Amat}{\Xi}

\newcommand{\buyers}{aggregators}
\newcommand{\buyer}{aggregator}

\usepackage{amsthm}
\usepackage{cite}
\newtheorem{assumption}{Assumption}
\newtheorem{definition}{Definition}

\newtheorem{lemma}{Lemma}
\newtheorem{theorem}{Theorem}
\newtheorem{corollary}{Corollary}

\newcommand{\PoA}{\mathrm{PoA}}
\newcommand{\Z}{\mc{Z}}

\title{\LARGE \bf
Competitive Statistical Estimation with Strategic Data
Sources}

\author{Tyler Westenbroek,~\IEEEmembership{Student Member,~IEEE}, Roy Dong,\\ Lillian J.
    Ratliff,~\IEEEmembership{Member,~IEEE}, and S. Shankar
    Sastry,~\IEEEmembership{Fellow,~IEEE}
\thanks{T.~Westenbroek and S.~S.~Sastry are with the Department of
    Electrical Engineering and Computer Sciences, University of California,
    Berkeley, Berkeley, CA, 94707, USA, e-mail:~{\tt\footnotesize
        $\{$westenbroekt,sastry$\}$@eecs.berkeley.edu}.}
\thanks{R.~Dong is with the Department of Electrical and Computer Engineering, University of Illinois at Urbana-Champaign, Champaign, IL, 61820, USA, e-mail:~{\tt\footnotesize
        roydong@illinois.edu}.}
\thanks{L.~J.~Ratliff is with the Department of Electrical Engineering,
            University of Washington, Seattle, WA, 98195, USA,
            e-mail:~{\tt\footnotesize ratliffl@uw.edu}.}
\thanks{This work is partially funded by NSF CNS:1656873.}%
}

\begin{document}

\maketitle
\thispagestyle{empty}
\pagestyle{empty}

%%%%%%%%%%%%%%%%%%%%%%%%%%%%%%%%%%%%%%%%%%%%%%%%%%%%%%%%%%%%%%%%%%%%%%

\begin{abstract}
In recent years, data has played an increasingly important role in the economy as a good in its own right. 
In many settings, data aggregators cannot directly verify the quality of the data they purchase, nor the effort exerted by data sources when creating the data. 
Recent work has explored mechanisms to ensure that the data sources share high quality data with a single data aggregator, addressing the issue of moral hazard. 
Oftentimes, there is a unique, socially efficient solution.

In this paper, we consider data markets where there is more than one data aggregator. Since data can be cheaply reproduced and transmitted once created, data sources may share the same data with more than one aggregator, leading to free-riding between data aggregators. 
This coupling can lead to non-uniqueness of equilibria and social inefficiency. 
We examine a particular class of mechanisms that have received study recently in the literature, and we characterize all the generalized Nash equilibria of the resulting data market.
We show that, in contrast to the single-aggregator case, there is either infinitely many generalized Nash equilibria or none. We also provide necessary and sufficient conditions for all equilibria to be socially inefficient. In our analysis, we identify the components of these mechanisms which give rise to these undesirable outcomes, showing the need for research into mechanisms for competitive settings with multiple data purchasers and sellers.
\end{abstract}

%%%%%%%%%%%%%%%%%%%%%%%%%%%%%%%%%%%%%%%%%%%%%%%%%%%%%%%%%%%%%%%%%%%%%%
%%%%%%%%%%%%%%%%%%%%%%%%%%%%%%%%%%%%%%%%%%%%%%%%%%%%%%%%%%%%%%%%%%%%%%
%%%%%%%%%%%%%%%%%%%%%%%%%%%%%%%%%%%%%%%%%%%%%%%%%%%%%%%%%%%%%%%%%%%%%%
\section{Introduction}
\label{sec:intro}
\input{intro}

\section{Related Literature}
\label{sec:background}
\input{back}

%%%%%%%%%%%%%%%%%%%%%%%%%%%%%%%%%%%%%%%%%%%%%%%%%%%%%%%%%%%%%%%%%%%%%%
%%%%%%%%%%%%%%%%%%%%%%%%%%%%%%%%%%%%%%%%%%%%%%%%%%%%%%%%%%%%%%%%%%%%%%
%%%%%%%%%%%%%%%%%%%%%%%%%%%%%%%%%%%%%%%%%%%%%%%%%%%%%%%%%%%%%%%%%%%%%%

\section{Data Market Preliminaries}
\label{sec:datam_formulation}
\input{formulation_head}

\subsection{Strategic Data Sources}
\label{sec:source_details}
\input{formulation_sources}

\subsection{Strategic Data Aggregators}
\label{sec:buyer_details}
\input{formulation_buyers}

\subsection{Structure of Payment Contracts}
\label{sec:payment_struct}
\input{payments}

\subsection{Formulation of Aggregator Optimization Problem}
\label{subsec:aggform}
\input{aggregator_formulation}

\subsection{Induced Equilibrium Between Data Sources}
\label{subsec:sourcegame}
\input{inducedgame}

\subsection{Reformulation of Buyers Optimization Problem}
\label{subsec:reformulation}
\input{reformulation}

\input{notationtable}

%%%%%%%%%%%%%%%%%%%%%%%%%%%%%%%%%%%%%%%%%%%%%%%%%%%%%%%%%%%%%%%%%%%%%%
%%%%%%%%%%%%%%%%%%%%%%%%%%%%%%%%%%%%%%%%%%%%%%%%%%%%%%%%%%%%%%%%%%%%%%
%%%%%%%%%%%%%%%%%%%%%%%%%%%%%%%%%%%%%%%%%%%%%%%%%%%%%%%%%%%%%%%%%%%%%%

\section{Generalized Nash Equilibria in the Data Market}
\label{sec:datam_results}
\input{gnep}

\subsection{Unbounded Effort Spaces}
\label{subsec:unbounded}
\input{unbounded}

\subsection{Bounded Effort Spaces}
\label{subsec:bounded}
\input{bounded}

\subsection{Conditions for Social Inefficiency}
\label{subsec:poa}
\input{poa}

%%%%%%%%%%%%%%%%%%%%%%%%%%%%%%%%%%%%%%%%%%%%%%%%%%%%%%%%%%%%%%%%%%%%%%
%%%%%%%%%%%%%%%%%%%%%%%%%%%%%%%%%%%%%%%%%%%%%%%%%%%%%%%%%%%%%%%%%%%%%%
%%%%%%%%%%%%%%%%%%%%%%%%%%%%%%%%%%%%%%%%%%%%%%%%%%%%%%%%%%%%%%%%%%%%%%

\section{Partial Data Sharing}
\label{subsec:partial}
\input{partial}

%%%%%%%%%%%%%%%%%%%%%%%%%%%%%%%%%%%%%%%%%%%%%%%%%%%%%%%%%%%%%%%%%%%%%%
%%%%%%%%%%%%%%%%%%%%%%%%%%%%%%%%%%%%%%%%%%%%%%%%%%%%%%%%%%%%%%%%%%%%%%
%%%%%%%%%%%%%%%%%%%%%%%%%%%%%%%%%%%%%%%%%%%%%%%%%%%%%%%%%%%%%%%%%%%%%%

%\section{Example: Between two firms}
%\label{sec:datam_examples}
%\input{example}

%%%%%%%%%%%%%%%%%%%%%%%%%%%%%%%%%%%%%%%%%%%%%%%%%%%%%%%%%%%%%%%%%%%%%%
%%%%%%%%%%%%%%%%%%%%%%%%%%%%%%%%%%%%%%%%%%%%%%%%%%%%%%%%%%%%%%%%%%%%%%
%%%%%%%%%%%%%%%%%%%%%%%%%%%%%%%%%%%%%%%%%%%%%%%%%%%%%%%%%%%%%%%%%%%%%%

\section{Closing remarks}
\label{sec:datam_discussion}
\input{conclusion}

%%%%%%%%%%%%%%%%%%%%%%%%%%%%%%%%%%%%%%%%%%%%%%%%%%%%%%%%%%%%%%%%%%%%%%
%%%%%%%%%%%%%%%%%%%%%%%%%%%%%%%%%%%%%%%%%%%%%%%%%%%%%%%%%%%%%%%%%%%%%%
%%%%%%%%%%%%%%%%%%%%%%%%%%%%%%%%%%%%%%%%%%%%%%%%%%%%%%%%%%%%%%%%%%%%%%

\appendix
\section{Generalized Nash Equilibrium Existence Theorem}
\label{sec:gne_exist}
\input{existence_thm}

\bibliographystyle{IEEEtran}
\bibliography{root}

\begin{IEEEbiography}[{\includegraphics[width=1in,height=1.25in,clip,keepaspectratio]{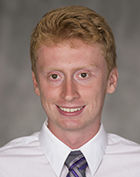}}]{Tyler~Westenbroek} is a graduate student at the University of California, Berkeley, pursuing a Ph.D. in Electrical Engineering and Computer Science. He graduated, with top honors, from Washington University in Saint Louis, receiving a B.S.~(2016) with majors in Systems Engineering and Computer Science. His interests lie primarily in the areas of Hybrid Systems, Optimal Control, and Optimization.
\end{IEEEbiography}
\begin{IEEEbiography}[{\includegraphics[width=1in,height=1.25in,clip,keepaspectratio]{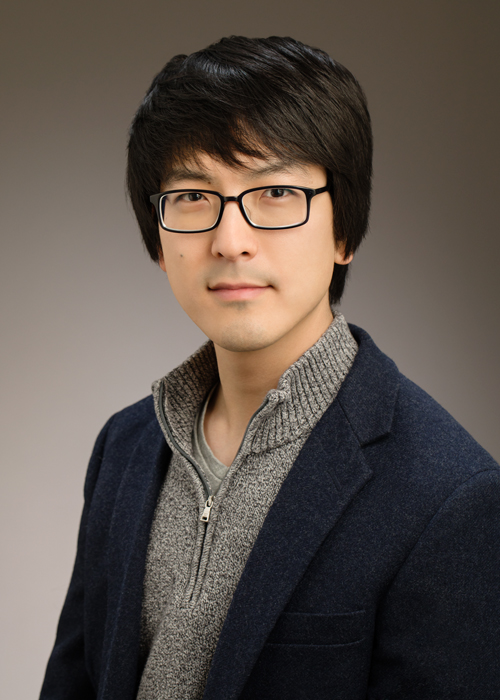}}]{Roy~Dong}
    is a Research Assistant Professor in the Department of Electrical and Computer Engineering at the University of Illinois at Urbana-Champaign. He was a postdoctoral researcher and visiting lecturer at University of California, Berkeley from 2017 to 2018, where he also received his Ph.D. in Electrical Engineering and Computer Sciences in 2017. Prior to his graduate studies, he received a B.S. Honors in Economics and a B.S. Honors in Computer Engineering from Michigan State University in 2010. He is the recipient of the National Science Foundation Graduate Research Fellowship (2011).
\end{IEEEbiography}
\begin{IEEEbiography}[{\includegraphics[width=1in,height=1.25in,clip,keepaspectratio]{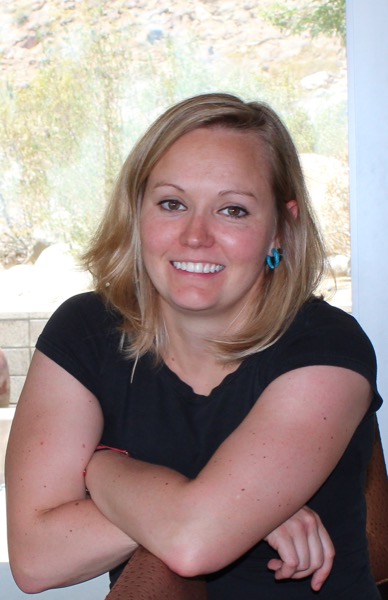}}]{Lillian
    J.~Ratliff} (S'08--M'15)
 is an Assistant Professor in Electrical Engineering (EE) at the University of Washington, Seattle. Prior
to joining UW she was a Postdoctoral Researcher in Electrical Engineering and
Computer Sciences at the University of California, Berkeley
where she also obtained her Ph.D.~in 2015. She obtained a B.S.~in Mathematics
(2008) and a B.S.~(2008) and M.S.~(2010) in EE all from the University of Nevada, Las Vegas.  Her research interests lie at the intersection of game theory, optimization, and
learning. She is the recipient of the National Science
Foundation Graduate Research Fellowship (2009) and the CISE Research Initiation
Initiative Award (2017).
\end{IEEEbiography}
\begin{IEEEbiography}[{\includegraphics[width=1in,height=1.25in,clip,keepaspectratio]{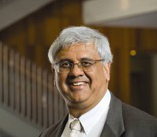}}]{S.~Shankar
    Sastry} (S'79--M'80--SM'95--F'95)
is with the University of California, Berkeley where he is faculty
director of the Blum Center for Developing Economies. He served as the Dean of Berkeley's College of Engineering from 2007-2018. He received his B.~Tech
from the Indian Institute of Technology, and M.S., M.A.~(Math), Ph.D.~in
Engineering from Berkeley.  He has served on the faculties of MIT and Harvard.
His areas of personal research are design of resilient network control systems,
 autonomous systems, computer vision,
nonlinear and adaptive control, and hybrid and embedded systems. 
He is also a member of the National Academy of Engineering and the
American Academy of Arts and Sciences. He received an honorary M.A.~from
Harvard and honorary doctorates from the Royal Swedish Institute of
Technology and the University of Waterloo.  He has been a member of the Air
Force Scientific Advisory Board and the Defense Science Board, among other
national boards. He is a member of the UN Secretary General's Scientific Advisory Board. 
He has coauthored over 550 technical papers and 9 books.  He has supervised over 85 doctoral students and over 50 MS students. 
His students now occupy leadership roles in several places and on the faculties of many major universities in the United States and abroad.
\end{IEEEbiography}
\end{document}

%% file: intro.tex
% !TEX root = root.tex

\IEEEPARstart{D}{ata} has increasingly seen a role in the economy as an important good. As an input to machine learning algorithms, data can not only create new products and innovations, but also be used to redesign business strategies and processes. As the demand for data increases, we have seen the formation of data aggregators, who collate data for either use or resale. A fundamental information asymmetry arises between data aggregators and data sources: how can aggregators verify the quality of the data they purchase from data sources?

In particular, data sources often incur an effort cost to obtain high quality data. 
For example, devices require maintenance and upkeep to ensure accurate measurements, 
portable sensors need to use their limited energy resources to collect and transmit data, 
and human agents may need to be compensated to properly perform a desired task. 
As such, if a data aggregator wants a high quality data point, they must appropriately compensate the data source. Furthermore, this problem is complicated by the fact that the data aggregators cannot observe the effort exerted, and only the data received. As such, the payments must be calculated from the data sets alone, with no knowledge of the effort exerted or noise levels of data points. This problem has led to the design of a variety of mechanisms to ensure data sources provide quality data, which we will outline in more detail in Section~\ref{sec:background}.

The contribution of this paper is the study of the data market that forms when multiple data aggregators share the same pool of data sources. In particular, we note that data is {non-rivalrous}, in the sense that it can be cheaply copied and shared with multiple data aggregators. Since a data aggregator does not `consume' the good after purchasing it, data sources will have an incentive to share the same data with as many aggregators as are willing to pay.
We show that the non-rivalrous nature of data introduces a coupling between data buyers: when a data aggregator incentivizes a data source to produce high quality data, other data aggregators benefit. 
In particular, this coupling leads to undesirable properties of the equilibrium. In many single-aggregator formulations, equilibria are unique and there is no social inefficiency. In contrast, the multiple-aggregator case leads to a multiplicity of equilibria, and social inefficiencies across all equilibria.

The rest of this paper is organized as follows. In Section~\ref{sec:background}, we discuss the related literature and contextualize our contributions. In Section~\ref{sec:datam_formulation}, we introduce our model for data sources, data aggregators, and their interactions in the data market. In Section~\ref{sec:datam_results}, we characterize the generalized Nash equilibria in the data market, and identify necessary and sufficient conditions for social ineffiency. In Section~\ref{subsec:partial}, we extend the results to cases where data sources do not share their data with all data aggregators. 
Finally, we close with final remarks in Section~\ref{sec:datam_discussion}.

%% file: back.tex
% !TEX root = root.tex

In recent years, there has been a quickly growing body of literature on models for data exchange and data markets. 
Broadly speaking, the existing literature can be broken down by two categories: models with a single data purchaser and single data source, and models with a single data purchaser and multiple data sources.

In the first category, we find a class of models which study a single data purchaser and a single data source. These works focus on the game theoretic interactions and information states between the two agents. In particular, these works consider the strategies arising from direct signals, actions, and payments, rather than indirect coupling that can arise from multiple sources or purchasers. 
Some of these papers feature multiple data sources, but these are ultimately separable into a collection of single-source models, and, at their core, focus on the direct interactions between buyers and sellers of data.  In~\cite{babaioff2012optimal}, optimal mechanisms for a single data source to sell to a single buyer are developed using a signaling framework. 
The authors of \cite{bergemann2018design} design a menu of prices for different data qualities, employing a screening framework. 
In~\cite{dobakhshari2017reputation}, the authors consider a single {\buyer} and single source, and show how repeated interactions with noisy verification allow for mechanisms which elicit costly effort from a data source.  A single data source charging data purchasers for queries about customer preferences is studied in \cite{bergemann2015selling}. 

In the second category, there are a class of models which study a single data purchaser with multiple data sources. These works focus on capturing how the data supplied by one data source affects another. In~\cite{caragiannis2016truthful}, the authors consider a single data {\buyer} and multiple data sources, and show how robustness of the sample median provides protection against strategic data sources. 
In~\cite{chen2018optimal}, the authors consider a single data {\buyer} and multiple data sources in a setting with verifiable data, and allow the data and the cost of revealing data to be arbitrarily correlated.

There is also a new body of work in the single-{\buyer}, multiple-source case, using \emph{peer prediction} mechanisms, first introduced in~\cite{miller2005eliciting}. These techniques often use scoring techniques to evaluate the `goodness' of received data, and often examine classification tasks. 
In~\cite{prelec2004bayesian,faltings2014incentive}, the authors develop mechanisms for eliciting the truth in crowdsourcing applications, while ~\cite{dasgupta2013crowdsourced,radanovic2015incentive, shnayder2016informed} consider theoretical extensions to strengthen the original results of~\cite{miller2005eliciting}, all in the context of a single {\buyer}. 
In~\cite{liu2017machine}, the authors consider a classification problem with a single {\buyer} and multiple data sources, which extends the classic peer prediction results by exploiting correlations between the queries and query responses.

A parallel literature considers similar ideas in the regression domain. These works design general payment mechanisms, by which a central data {\buyer} may incentivize data sources to exert the effort necessary to produce and report readings which are deemed to be of high quality, with respect to the estimation task the {\buyer} is performing. The roots of these approaches can be traced least as far back as VCG mechanisms, a set of seminal results in mechanism design~\cite{vickrey:1961aa}. Indeed, numerous approaches for deciding payments based on the actions of other agents have been proposed~\cite{Nisan2007}. Here, we again see attention given to crowdsourcing \cite{chawla:2015aa}. 

Several recent papers~\cite{cai:2015aa,farokhi2015budget,Farokhi2015,dobakhshari2016incentive,Farokhi2017,dobakhshari2017reputation} investigate new directions in this domain. In cases where, without the ability to directly determine the effort exerted by data sources, data buyers must design incentive mechanisms based solely on the data available to them. In \cite{cai:2015aa}, whose approach we extend here, the authors develop a mechanism which a data aggregator can use to precisely set the level of effort a collection of data sources exert when producing data. A similar mechanism is explored in \cite{farokhi2015budget}. Extensions are considered wherein data sources form coalitions \cite{Farokhi2015}, or where aggregators assess the quality of readings using a trusted data source \cite{dobakhshari2016incentive}.  Meanwhile, \cite{Farokhi2017} and \cite{dobakhshari2017reputation} investigate dynamic settings where data sources are repeatedly queried. 

Our work is closest in spirit to the literature studying regression problems with multiple data sources, with our key contribution being the presence of multiple data {\buyers} that are coupled in their costs and actions. 
To our knowledge, this is one of the first papers which considers multiple data {\buyers} and multiple data sources simultaneously. In particular, we simultaneously model coupling between data {\buyers} in their cost functions, coupling in the payments to the same pool of data sources, and coupling between data sources due to payments that depend on their peers' data. 

We suppose all data {\buyers} are trying to estimate the same function and share the same pool of data sources. Additionally, we assume each data {\buyer} has already chosen an estimator, and now must determine how to issue payments to have low estimation error with their exogenously fixed estimator. Our model builds heavily on the model introduced in~\cite{cai:2015aa}, which featured a single data {\buyer}. Our contribution is an extension that models cases with multiple data {\buyers}. 
For consistency, we will refer to data purchasers as \emph{data {\buyers}}, and data sellers as \emph{data sources}. 

Furthermore, the work in the paper is a significant extension of our prior
work~\cite{westenbroek:2017aa} where we considered strategic data sources with a
specific exponential function mapping effort to query response quality. In the
present work, we characterize equilibria and
the price of anarchy for a much broader class of games between data buyers where
the data sources' effort functions can be any non-negative, strictly decreasing,
convex, and twice continuously differentiable function. The characterization we
provide considers both bounded and unbounded feasible effort sets for the data
sources.

%% file: formulation_head.tex
% !TEX root = root.tex

In this section, we outline the models for data sources, data {\buyers}, and the strategic interactions between them. 

At a high level, each data {\buyer} collects data from data sources to construct an estimate of a given function. In exchange for this data, the data {\buyer} issues incentives to the data sources. The data {\buyers} have three terms in their cost function:
1) an estimation error term, which rewards the data {\buyer} for constructing a better estimate; 2) a competition term, which penalizes when other data {\buyers} have higher quality estimates; 3) a payment term, which is the cost incurred issuing incentives.

Each data source is able to produce a noisy sample of the desired function. The data sources can exert effort to reduce the variance of the data sample, and we assume the data sources are \emph{effort-averse}, i.e.  data sources will prefer to exert less effort, unless they are provided incentive by the {\buyers}. As such, the data sources have two terms in their utility function: 1) an incentive term, which rewards payments received; 2) an effort term, which penalizes effort exerted.

The level of effort exerted and the variance of the data are not known by the data {\buyer}; this \emph{private information} gives rise to \emph{moral hazard}. One of the problems for the {\buyer} is the task of designing incentives which depend only on the information available to them. Another important nuance is that data is \emph{non-rivalrous}; thus, when a data source produces a higher-quality data sample, all the {\buyers} which receive this data benefit.

In order to simplify the initial introduction of our model, we will first assume that each data source provides data to all the {\buyers} in the data market, and receives payment from all {\buyers} as well. In Section \ref{subsec:partial}, we will outline how our results change when this assumption is removed. 
\begin{comment}
Throughout this section, we will assume that all data sources sell their data to all {\buyers}. We will relax this assumption in Section~\ref{subsec:partial} and show how the conclusions change as a result.
\end{comment}

\subsection{Overview}

More formally, let $\mc{S} = \{1, \ldots, N\}$ be the index set of \emph{strategic data sources}, and let $\mc{B} = \{1,\ldots, M\}$ be the index set of \emph{strategic data {\buyers}}.
Each data {\buyer}
desires to construct an estimate for a given function $f \colon \mc{D} \to \R$, where $\mc{D}$ 
is a \emph{feature space}. 
Practically, one may think
of $\Dc$ as a set of features the data {\buyers} are capable of observing, while the mapping 
 $f$ encapsulates the relationship between the observable features and the outcome of interest.  

Each data
source $s \in \mc{S}$ is able to produce a noisy sample $y_s$ of $f$ at the
fixed point $x_s \in \mc{D}$. The point $x_s$ is common knowledge among all data sources and {\buyers}. 
The variance of $y_s$ is proportional to the
effort exerted by data source $s$ to produce the reading. Each data source $s$ is characterized by an \emph{effort-to-variance} function $\sigma_s^2 : \mc{E}_s \to \R_{\geq 0}$, where $\mc{E}_s$ represents the set of feasible efforts that data source $s$ can exert. When data source $s$ exerts effort $e_s \in \mc{E}_s$, they produce the data point:
\begin{equation}
\label{eq:data_sample}
y_s(e_s) = f(x_s) + \eps_s(e_s)
\end{equation}
Here, $\eps_s(e_s)$ is a random variable with mean $0$ and variance $\sigma_s^2(e_s)$. 
The function $\sigma_s^2$ is common knowledge among all data sources and {\buyers}. However, while the function $\sigma_s^2$ is known, the effort exerted $e_s$ is private. This means that the actual variance of $y_s$, namely $\sigma_s^2(e_s)$, is also private information of $s$. We will delve into assumptions in the data source model in greater detail in Section~\ref{sec:source_details}.

Now, suppose a data {\buyer} is granted access to a data set $\{(x_s,y_s)\}_{s \in S}$.  At this point, the data {\buyer} $b \in \mc{B}$ 
processes this data to construct an estimate for $f$. In exchange for this data set, the data {\buyer} issues payment $p_s^b(y)$ to data source $s$ for each $s \in \mc{S}$. Here, $y = (y_1,\dots,y_N)$ denotes the data 
given to each member of $\mc{B}$. Note that the payment to $s$ from $b$ depends not only on the data supplied by $s$, but rather depends on all data available to $b$. 

The data {\buyer} then incurs loss $L^b(p^b,e)$, which will depend on $p^b = (p_i^b)_{i \in \Sc}$, the payments issued, as well as $e = (e_i)_{i \in \Sc}$, the effort exerted by the data sources. 
We will formalize the data {\buyer} in greater detail in Section~\ref{sec:buyer_details}.

The interaction of the data market proceeds in three stages.

\begin{enumerate}

\item \emph{Aggregators declare incentives:} Each data {\buyer} $b \in \mc{B}$ commits to a payment contract $p^b = (p_i^b)_{i \in \Sc}$. The payments will depend on the data $y$ shared with $b$, as well as the common knowledge information $x = (x_i)_{i \in \Sc}$ and functions $\sigma^2 = (\sigma_i^2)_{i \in \Sc}$.

\item \emph{Sources exert effort, realize and share data:} In response to $p^b$, each data source $s$ chooses an effort $e_s \in \mc{E}_s$. Then, the random variable $y_s$ is realized according to~\eqref{eq:data_sample}. The data $y_s$ is shared with each data {\buyer}. 
Note that $\src$ has control over $y_s$ only through $e_s$. In other words, the data source chooses the quality of data they generate, but cannot arbitrarily manipulate the reported value of $y_s$. 

\item \emph{Aggregators construct estimates, issue payments:} Each data buyer $b$ constructs their estimate $\hat{f}^b$, issues payments $p^b$ to the data sources, and incurs loss $L^b$.

\end{enumerate}

For convenience, we include a table summarizing the notation throughout this paper in Table~\ref{tab:notation}.

%% file: formulation_sources.tex
% !TEX root = root.tex

As mentioned previously, each data source $s \in \mc{S}$ has their own \emph{feature
vector} $x_s \in \Dc$, and samples the function $f$ at this point. We may also refer to $x_s$ as a \emph{query} throughout the text, and $y_s$ as the \emph{query response} for data source $s$. The data source $s$ is characterized by the \emph{effort-to-variance} function $\sigma_s^2 : \mathcal{E}_s \rightarrow \R_{\geq 0}$. We assume $0 \in \mc{E}_s$ so that each data source may exert no effort in producing her reading if she desires. 

\begin{assumption}
\label{ass:Eset}
For each $s \in \mc{S}$, the set $\mc{E}_s \subset \R_{\geq 0}$ is a closed, connected set and contains $0$.
\end{assumption} 
Assumption~\ref{ass:Eset} means that we consider two cases:
\begin{enumerate}
    \item[(i)] $\mc{E}_s = [0, \infty)$, i.e.~the data sources maximum allowed
    effort is unbounded.
\item[(ii)] $\mc{E}_s = [0,
    e_s^{\max}]$ for some $0<e_s^{\max}<\infty$, i.e.~the data sources maximum allowed
    effort is bounded.
\end{enumerate}
Imposing an upper-bound on the amount of a effort a data source can exert can be used to model constraints such as hardware limitations. As we shall see in Section \ref{sec:datam_results}, the imposition of such constraints can drastically affect equilibrium behavior in the data market.

Once the data source $s$ exerts effort $e_s \in \mc{E}_s$, they produce the data point $y_s$ according to~\eqref{eq:data_sample}. 
Again, we note that the data source only controls the effort level $e_s$. They can only indirectly control $y_s$ through $e_s$, and cannot report arbitrary values as their data. We also impose the assumption that the noise in the data is independent across data sources.

\begin{assumption}
For each $s \in \mc{S}$, $\eps_s(e_s)$ is a random variable with mean $0$ and variance $\sigma_s^2(e_s)$. Furthermore, the random variables $\{\eps_s(e_s)\}_{s \in \mc{S}}$ are independent.
\end{assumption}

Both $x_s$ and the function function $\sigma_s^2$ are common knowledge, but the effort $e_s$ and $\sigma_s^2(e_s)$, the actual variance of $y_s$, are
private. 

For convenience, we let $\mathcal{E} = \mathcal{E}_1 \times \dots \times
\mathcal{E}_N$ be the joint effort set and let  $\sigma^2=(\sigma^2_1, \ldots,
\sigma_{N}^2)$ be the tuple of effort-to-variance functions. We make the following  assumptions on the effort-to-variance mappings
$\sigma^2$.

\begin{assumption}
\label{ass:sigma_form}
    For each data source $s\in \mc{S}$, the mapping
    $\sigma_s:\mathcal{E}_s\to \R_{\geq 0}$, which is the square root of
    $\sigma^2_s$, is
    (i) strictly decreasing,
(ii) convex, and (iii) twice continuously differentiable. 
\end{assumption}

The assumptions correspond to the variance of the estimate generated by data source
$s$ decreasing in the effort exerted, with
decreasing marginal returns. 

Using the notation $p_s = (p_s^j)_{j \in \Bc}$, we model each data source with the following utility function:
\begin{equation}
\label{eq:effort_selection}
\textstyle u_s(e_s, p_s)=
    \E  \Big( \sum_{j \in \mc{B}} p_s^j( y(e)) \Big) - e_s
\end{equation}
where the expectation is with respect to the randomness in $y$, the data generated by the data sources upon exerting effort $e$.\footnote{For simplicity and as a first-step analysis, we assume that the data sources only care about the payments received from the {\buyer}, and are indifferent to which {\buyers} they share their data with. An interesting and practical extension would be to consider the case where the data sources' utility functions are {\buyer}-dependent. This could arise when data sources trust different {\buyers} differently, or over privacy concerns.} Note the form of~\eqref{eq:effort_selection} implies that the data sources are risk-neutral and effort-adverse.  Additionally, the form of~\eqref{eq:effort_selection} also implies the effort $e_s$ can be normalized to be comparable to the payments. We note that the timing of the game implies that data sources must commit to an effort level ex-ante.

Thus, in the second stage of the game, data source $s$ has knowledge of the payment contracts $(p^b)_{b \in \mc{B}}$, and chooses $e_s$ to maximize their $u_s(e_s,p_s)$, defined by~\eqref{eq:effort_selection}. However, since the utility of each data source depends on the effort exerted by the other data sources, the payments $(p^b)_{b \in \mc{B}}$ induce a game between the data sources. In Section~\ref{subsec:sourcegame} we will fully characterize this game for the particular class of incentives we introduce in Section \ref{sec:payment_struct}.

\begin{comment}
Thus, in the second stage of the game, data source $s$ has knowledge of the payment contracts $p_s = (p_s^j)_{j \in \mc{B}}$. They choose effort $e_s$ to maximize their utility $u_s(e_s,p_s)$, defined by~\ref{eq:effort_selection}.
\end{comment}

%% file: formulation_buyers.tex
% !TEX root = root.tex

The primary objective of each {\buyer} is to construct a low-variance estimate for the function $f$. We adopt the following formal definition for an estimator. 
 
\begin{definition}[Estimator~\cite{cai:2015aa}]
    Let $\mc{H}$ be a family of functions $f:\mc{D}\to\mb{R}$. An estimator for $\mc{H}$ takes as input a
    collection $\mc{X}=(x_i,y_i)_{i=1}^N$ of examples $(x_i,y_i)\in \mc{D}\times
    \mb{R}$ and produces an estimated function $\hat{f}_{\mc{X}}\in \mc{H}$.
    \label{def:est}
\end{definition}

As an example,  $\mc{H}$
may be the class of linear functions $f:\R^n\rightarrow \R$, in which case one may produce an estimated function $\hat{f}_{\mc{X}}\in \mc{H}$ of $f$ via linear regression.

Each data {\buyer} $\buy\in \mc{B}$ constructs his estimate for $f$ from the class of functions $\mc{H}_b$, using the readings $\mc{X} = (x_s , y_s)_{s \in \mc{S}} $. We let $\hat{f}^\buy_{\mc{X}} \in \mc{H}_b$ denote the
estimate that {\buyer} $\buy$ constructs based on the readings they receive.\footnote{In general,
{\buyers} need not fit models of the same type---e.g., one data
{\buyer} may choose to generate their estimate via linear regression, while
another fits a polynomial of higher degree.  Different
estimator types across data {\buyers} may be used to encapsulate competitive
advantages one has over another.}

Each data {\buyer}'s estimator is given, fixed, and common knowledge among all agents. In other words, this means that, for each data {\buyer}, the process by which a data set is turned into an estimate is exogenous. We focus on the design of incentives once each buyer has chosen an estimator.

First, we introduce some restrictions on the class of estimators allowed. 
The following assumption is required for us to be able to consider the contribution of data source $s$ to reducing {\buyer} $b$'s estimation cost. Also, note that the functions $\{h_b\}_{b \in \mc{B}}$ will be non-negative by construction.
\begin{assumption}
\label{ass:estimator_sep}
We assume the estimator for each $b \in \mc{B}$ is {separable}, in the following sense~\cite{cai:2015aa}. There exists a function $h_b$ such that
for all
queries ${\bf x}$, distributions $F$ over
$\mathcal{D}$, and variances ${\sigma}^2$ 
of the reported estimates ${\bf y}$ at queries ${\bf x} = (x_i)_{i = 1}^k$ in the dataset $\mc{X} = ({\bf x},{\bf y})$:
\begin{equation}
\label{eq:h_def}
\E \left[ \big(\hat{f}^\buy_{\mc{X}}(x^\ast) - f(x^\ast) \big)^2 \right] = \sum_{i = 1}^k h_b(x_i,{\bf x},F) \sigma_i^2
\end{equation} 
Here, the expectation is taken across the randomness in $\mc{X}$, as well as across $x^\ast \sim F$.
\end{assumption}
For brevity, we will also define the function $g_\buy$ as follows:
\begin{equation}
\label{eq:g_def}
 \textstyle   g_\buy({\bf x}, F, {\sigma}^2) = \sum_{i = 1}^k h_b(x_i,{\bf x},F) \sigma_i^2
\end{equation} 

Let  $-\buy=\mc{B}\setminus\{\buy\}$ denote the index set of {\buyers}
excluding $\buy$ and let 
$p^{-\buy}=(p^j)_{j\in -b}$ be the payments of all
{\buyers} excluding $\buy$. 
Aggregator $\buy$ constructs
 payments so as to minimize:
\begin{align}
\label{eq:buyer_obj}
\textstyle L^\buy({p^\buy}, e)  &=\textstyle
\E \Big[ 
    \big( \hat{f}^\buy_{\mc{X}}(x^\ast) - f(x^\ast) \big)^2 \notag \\
    & \textstyle \quad -\sum_{j \in -\buy} \comp_j^\buy \big( \hat{f}^j_{\mc{X}}(x^\ast) - f(x^\ast) \big)^2 \\
    & \textstyle \quad + \eta^\buy \sum_{\src \in \mc{S}} p_\src^\buy(y(e)) \notag
\Big]
\end{align}
As in~\eqref{eq:h_def}, the expectation in~\eqref{eq:buyer_obj} is taken with respect to $x^\ast \sim
F_\buy$ and the randomness in the query responses $y$. 
The distribution $F_\buy$ weighs the importance data {\buyer} $\buy$ places on accurately
estimating $f$ for different query points $x \in \mathcal{D}$. 

The scalars $\comp_j^\buy \in [0 ,1]$ parameterize the level of competition between
{\buyers} $\buy$ and $j$. 
When $\comp_j^\buy =0$, {\buyer} $\buy$ is indifferent to the
success of $j$'s estimation; $\buy$ interacts with $j$ entirely through the incentives issued to the data sources. We note that, even when $\comp_j^\buy = 0$ for all $j$ and $b$, we can still see degeneracies and social inefficiency arise, since data aggregators will still be coupled through the data sources.\footnote{This is a stylized formulation of how competition can affect different data {\buyers}, but we see interesting results arise even in this simple model. In the future, we hope to consider more extensive models of competition for data {\buyers}.} The parameter $\eta^\buy > 0$ denotes a conversion between dollar amounts allocated by the payment functions and the utility generated by the quality of the various estimates that are constructed. We make the assumption that {\buyer} $\buy$ has knowledge of what estimator every other data {\buyer} plans to use, as well as the weighting distributions.\footnote{This is a fairly strong assumption given that competing data {\buyers} are unlikely to inform their competitors how they intend to process the data supplied by the sources. Our work isolates how coupling between {\buyers} through data sources affect the data market; an interesting avenue for future work is to consider extensions with different information sets, and characterize the existence and severity of market inefficiencies in these various situations.}

%% file: payments.tex
% !TEX root = root.tex

Throughout this paper, we will assume a particular form for the payment contracts the {\buyers} offer to the data sources. Similar to previous notation, we let $-s = \mc{S} \setminus \{s\}$. For a given $b \in \mc{B}$ and $s \in \mc{S}$ we assume that $p_s^b$ is of the form:
\begin{equation}
    p_\src^\buy({y}^\buy) = c_\src^\buy - \dpara_\src^\buy \left(
    y_\src^\buy -
    \hat{f}^\buy_{\mc{X}_{-\src}}(x_\src) \right)^2
    \label{eq:paystructure}
\end{equation}
Here, $a_s^b$ and $c_s^b$ are nonnegative scalars. Also,
$\mc{X}_{-\src}=(x_{-\src}, y_{-\src})$ denotes $b$'s data set excluding $s$. Namely $x_{-s}$ is the data features for 
all sources excluding $\src$ and $y_{-\src}=(y_i)_{i\in -s}$ is the query responses to {\buyer} $\buy$, excluding $s$.

    Note that these payments do not directly depend on the level of effort that any
of the data sources exert, since the data {\buyers} do not have a means to
directly observe these values. Rather, the payment to source $s$ from {\buyer}
$b$ depends on the $b$'s best estimate for $f(x_s)$ excluding $s$'s data, namely, $\hat{f}^\buy_{\mc{X}_{-\src}}(x_\src)$. 
The payments only depend on the data reported to them, and can be calculated by the {\buyer}.

Similar payment contracts are common in the literature~\cite{cai:2015aa,farokhi2015budget,dobakhshari2016incentive}, in part because of their intuitive structure. The {\buyer} constructs an unbiased estimate of what data source $s$ \emph{should} report, and this estimate is not influenced by the data of $s$. This estimate is used to overcome the problem of moral hazard: all data sources are appropriately incentivized to reduce the variance of their reported data accordingly.

Given this payment structure, each data {\buyer}'s choice of payment 
contracts reduces to choosing parameters 
$(c^\buy, a^{\buy})$ where $c^\buy =(c_i^{\buy})_{i\in\mc{S}}\in \mb{R}^{N}$ and $a^\buy=(a_i^\buy)_{i\in \mc{S}}\in \mb{R}^N$.

In the single {\buyer} case (when $M=|\mc{B}|=1$), it was shown in~\cite{cai:2015aa} that payments of the form in~\eqref{eq:paystructure} induce
a game between the data sources for which there is a unique dominant strategy 
equilibrium. That is, for each collection of parameters $(c_i^b)_{i \in \mc{S}}$ and
$(a_i^b)_{i \in \mc{S}}$, the data sources each exert a unique level of effort. The authors develop and algorithm by which the single {\buyer} may select these parameters such that
(i) data sources are incentivized to exert any level of effort that the
    {\buyer} desires, and 
(ii) data sources are compensated at exactly the value of their effort, i.e.~$\E [p_\src(y(e))] = e_\src$.

This paper's contribution is the study of how pricing schemes of this form
perform in the more general case where there is more than one data
{\buyer} (when $M = |\mathcal{B}| > 1$), and data {\buyers} may compete with each other. The goal is to model 
multiple {\buyers} as
strategic decision-makers in competition, and understand the \emph{data market} where these agents interact. Thus, while prior work captured moral hazard, we extend this model to capture competition and the non-rivalrous nature of data.

%% file: aggregator_formulation.tex
% !TEX root = root.tex

As mentioned previously, the {\buyers} hope to minimize their costs, as given in~\eqref{eq:buyer_obj}. They do so by choosing the parameters $(c^b, a^b)$. In this section, we will describe the {\buyer}'s optimization problem in more detail, and specify constraints that the parameter choice must satisfy.

The first constraint is \emph{individual rationality (IR)}. Individual rationality requires that 
each data source's utility is non-negative {ex-ante}~\cite{bolton:2005aa}.\footnote{Alternatively, a data
source's utility may be compared to an outside option; for
simplicity, we model the outside option as having zero utility.} This ensures that rational data sources are willing to exert effort to produce the data. The second constraint is non-negative payments from each data {\buyer}. Given that there are multiple {\buyers}, we introduce a constraint that the payment each {\buyer} offers to each ${\src}$ is non-negative {ex-ante}.\footnote{Negative payments could be
handled via exchangeable utilities among the data {\buyers} or via a trusted third--party to manage the allocations; however, in an
effort to ensure clarity, we leave these scenarios aside.}

We'll introduce some notation for brevity here; we let $\mathrm{p}_\src^\buy$ denote the expected value of the payment $p_s^b$:
\begin{align}
\mathrm{p}_\src^\buy& ( (c_\src^b, \dpara_\src^b), e ) 
= \E [ p_s^b(y) ] \notag \\
& =
c_\src^\buy-\dpara_\src^\buy\big(
\sigma_\src^2(e_\src) 
+ g_\buy(x_{-\src}, \delta_{x_\src},
\sigma^2_{-\src}(e_{-\src}))\big)
\label{eq:exppay}
\end{align}
where $\delta_x$ denotes the probability measure with mass one at $x$ and $e = (e_i)_{i \in \Sc}$. Similar to previous conventions, we define:
\begin{equation*}
    \mathrm{p}_\src((c_\src,\dpara_\src),e)=\textstyle\sum_{b\in
        \mc{B}}\mathrm{p}_\src^\buy((c_\src^b,\dpara_\src^b),e)
    \label{eq:sumtotal}
\end{equation*}
Thus, the IR constraint for each data source $s$ is formalized:
\begin{equation}
\label{eq:IR_conditions}
 \mathrm{p}_\src((c_\src,\dpara_\src),e)
 \geq e_\src
\end{equation}
Similarly, the non-negativity constraint for each data source $s$ and data {\buyer} $b$ is given by:
\begin{equation}
\label{eq:IR_conditions2}
\mathrm{p}_\src^\buy((c_\src,\dpara_\src),e)\geq 0
\end{equation}

The third constraint is \emph{incentive compatibility (IC)}. 
Intuitively, IC states that when a data source is acting rationally and choosing actions to maximize their utility, they behave as the data {\buyers} intended. 
When there is a single {\buyer}, IC is typically
enforced by the {\buyer} finding the effort that minimizes their cost, $e^\ast_\src$, and then
designing $p_\src$ such that
$e^\ast_\src = \arg\max_{e_s\in \mc{E}_\src} \mathrm{p}_\src((c_s,a_s),e)-e_s$.\footnote{For notational brevity, we will use $\arg\max$ as a function rather than a set-valued function throughout this paper; this is well-defined by Assumption~\ref{ass:sigma_form}.}

In the competitive
setting, IC for one aggregator is defined holding all other
aggregators payments fixed. Each of the data {\buyers} make their choice of payment subject to the
fact that data source 
$\src$ selects effort according to
\begin{equation}
\textstyle\max_{e_\src \in \mathcal{E}_\src} \sum_{j \in \Bc} \mathrm{p}_s^j((c_s^j,a_s^j),e) - e_\src
    \label{eq:IC}
\end{equation}
Note that the payment each source receives depends on the efforts exerted by the other data sources. Thus, for each set of contracts offered by the {\buyers}, a game is induced between the data sources to determine how much effort they will exert. The {\buyers} compete by issuing incentives, which influences the equilibrium behavior of this game. 

From the perspective of the data {\buyers}, the IC constraint states the desired effort level $e_s^*$ must be a \emph{dominant strategy} for data source $s$; that is, $e_s^*$ is the utility-maximizing action for $s$ regardless of the actions taken by other sources $-s$. Formally, the following must hold for all $e_{-s} \in \mc{E}_{-s}$:
\begin{equation*}
e_\src^\ast= \arg\max_{e_\src \in \mathcal{E}_\src}
\mathrm{p}_\src((c_\src,\dpara_\src),(e_\src, e_{-\src}))-e_\src
\end{equation*}

With these constraints, we formulate a bilevel optimization problem for each {\buyer}. Consider a fixed {\buyer} $b \in \mc{B}$. Given a fixed action profile for all other buyers $-b$, i.e.~given $(c^{-b},a^{-b})$, {\buyer} $b$ aims to solve:
\begin{align}
    \min_{(c^\buy, \dpara^\buy)} &\ L^\buy( (c^\buy,
    \dpara^\buy), (c^{-\buy}, \dpara^{-\buy}) ) \notag \\
\text{s.t.}\ \ &\ \
e_\src^\ast= \arg\max_{e_\src \in \mathcal{E}_\src}
\mathrm{p}_\src((c_\src,\dpara_\src),(e_\src, e_{-\src}))-e_\src,\notag\\
&\ \ \quad \forall
e_{-\src}\in \mc{E}_{-\src}, \ \forall s\in \mc{S}
\notag \\
&\ \mathrm{p}_\src((c_\src,\dpara_\src),(e^\ast_\src, e_{-\src}^\ast))\geq
e_\src^\ast,\ \forall s\in \mc{S}
\notag \\ 
&\  \mathrm{p}_\src^\buy((c_\src,\dpara_\src),(e_\src^\ast, e_{-\src}^\ast))\geq
0, \ \forall s\in \mc{S}
\notag \\
&\ c_\src^\buy \geq 0, \dpara_\src^\buy \geq 0, \ \forall s\in \mc{S} \notag
\end{align}
where $L^\buy$ is defined in~\eqref{eq:buyer_obj}. 

Note that this problem actually has $N$ optimization problems as constraints, making is a difficult bilevel program. However, we will reformulate the {\buyer}'s problem to a more manageable non-linear program in the sequel. This is possible, in part, due to the nice properties of the payment contract structure introduced in Section~\ref{sec:payment_struct}; this tractability motivates the use of payment contracts of that particular form. Next, we analyze the induced game between the data sources and simplify the {\buyer}'s optimization problem.

%% file: inducedgame.tex
% !TEX root = root.tex
To ensure a notion of \emph{incentive compatibility in equilibrium}, we show there is a well-defined mapping from the  parameters $(c,a)$ chosen by the {\buyers} to the equilibrium $e^\ast$. 

\begin{definition}
For fixed payments $\{p_\src^\buy\}_{\src\in \mc{S}, \buy\in \mc{B}}$, we say $e^\ast=(e_1^\ast, \ldots, e_N^\ast)$ is an \emph{induced Nash equilibrium} if for each
data source $\src \in \mc{S}$:
\begin{equation}
\textstyle e_\src^\ast=\arg\max_{e_\src \in \mathcal{E}_\src} \E \left[ \sum_{j \in \Bc} p_\src^j ({y}({e_\src, e_{-\src}^\ast})) \right] - e_\src
\label{eq:nashsources}
\end{equation}
If~\eqref{eq:nashsources} holds for all $e_{-s} \in \mc{E}_{-s}$ rather than just at $e_{-s}^*$, then we say that $e^*$ is an \emph{induced dominant strategy equilibrium}.
\end{definition}

Suppose now that we have a set of payments of the form discussed in Section~\ref{sec:payment_struct}, characterized by parameters $(c,a)$. Data source $s$ chooses effort $e_s^*$ according to:
\begin{multline}\label{eq:effort1}
\medmuskip=-1mu
\thinmuskip=-1mu
\thickmuskip=-1mu
\nulldelimiterspace=-1pt
\scriptspace=0pt
e_s^* = \arg \max_{e_s\in \mc{E}_\src} \Bigg[ \sum_{b \in \mc{B}}c_\src^\buy -\dpara_\src^\buy\big(
\sigma_\src^2(e_\src) 
+ g_\buy(x_{-\src}, \delta_{x_\src},
\sigma^2_{-\src}(e_{-\src}))\Bigg]  - e_s
\end{multline}
for each choice of $e_{-s} \in \mc{E}_{-s}$ made by the other data sources. It is straight forward to verify that \eqref{eq:effort1} is a concave maximization problem which admits a unique globally optimal solution. This follows from our assumption that $\sigma_s^2$ is convex and decreasing, recalling that $a_s^b \geq0$ for each $b \in \mc{B}$ and observing that $\mc{E}_{s}$ is a convex set. Moreover, note that the choice of this optimal effort $e_s^*$ is not affected by the choice of $e_{-s}$, since each of the $g_\buy(x_{-\src}, \delta_{x_\src},
\sigma^2_{-\src}(e_{-\src}))$ terms enters \eqref{eq:effort1} as a constant from the perspective of $s$. Thus, each choice of contract parameters selected by the {\buyers} leads to an induced dominant strategy equilibrium for the data sources. In particular, note that the choice of
\begin{equation}\label{eq:bara}
\textstyle\bs{\dpara}_\src = \sum_{j \in \mathcal{B}} \dpara_\src^j
\end{equation}
fully characterizes the level of effort that data source $s$ exerts in equilibrium. We reiterate that the constraints on the {\buyer}'s optimization problems will ensure  the chosen contract parameters respect the IR and non-negativity constraints. 

Next, we define $\mu_\src \colon \R_{> 0} \to \R_{\geq 0}$ to be the implicitly-defined map such that
$\mu_\src:\bs{\dpara}_\src\mapsto e_\src^*$ returns the solution to
\eqref{eq:effort1} for 
a given choice of $\bs{\dpara}_\src\in \mb{R}_{> 0}$. In the following section, we will use this mapping to simplify the optimization problem facing each of the {\buyers}. 

\begin{definition}\label{def:feasible_d}
For a given data source $\src\in \mc{S}$, let:
\begin{equation}\label{eq:dmax}
\bs{\underline{{\dpara}}}_\src=\min\ 
\{\bs{\dpara}_\src\in \mb{R}_{>0} :
\mu_\src({\bs{\dpara}}_\src) = 0\}.
\end{equation}
When $\mathcal{E}_s =
[0 , e_\src^{\max}]$ with $0\leq e_\src^{\max}<\infty$, define 
${\mc{A}_\src = [\bs{\underline{\dpara}}_\src, \bs{\bar{\dpara}}_\src]}$ where
\begin{equation}\label{eq:dmin}
\bs{\bar{\dpara}}_\src=\min\ \{\bs{\dpara}_\src\in \mb{R}_{>0} : 
\mu_\src(\bs{\dpara}_\src)=e_\src^{\max}\}
\end{equation}
On the other hand, when $\mathcal{E}_\src = \R_{\geq 0}$, define $\mc{A}_\src =
[\bs{\underline{\dpara}}_\src, \infty)$.
\end{definition}
The above definition implies $\bs{\underline{\dpara}}_\src$ is the minimum value
of $\bs{a}_s$ that the {\buyers} must offer data source $s$ to ensure they do
not have incentive to exert negative effort.\footnote{This situation could correspond to
source $s$ obfuscating their data, for example. We have restricted $\mc{E}$ to the non-negative orthant, so we will add constraints to ensure we are operating within the domain of our model.} Similarly, if the {\buyers} increase
$\bs{a}_\src$ past $\bs{\bar{\dpara}}_\src$, source $s$ cannot further increase
the level of effort they exert, and the mapping $\mu_s$ ceases to be meaningful. Thus, when reformulating each buyers optimization in the following section we will additionally constrain $\bs{a}_s \in \mc{A}_s$ for each $s \in \mc{S}$. 

The following lemma provides properties on the mapping $\mu_\src$ which are
needed to prove existence of equilibria for the game between {\buyers} in the first stage.

\begin{lemma}\label{lemma:increasing_effort}
    Fix a data source $\src\in \mc{S}$. Then the mapping $\mu_\src(\bs{\dpara}_\src)$ is
    continuous and strictly increasing in $\bs{\dpara}_\src$ for all values of
    $\bs{\dpara}_\src \in \mc{A}_s$.
\end{lemma}

\begin{proof}
The first-order optimality condition for
the data source is given by:
\begin{equation}\label{eq:effort_d}
\textstyle   2\bs{\dpara}_\src \sigma_\src(e_\src)\frac{d}{de_\src}\sigma_\src(e_\src) +1=0
    %\frac{1}{2\bs{\dpara}_\src}=0.
\end{equation}
By assumption  $\sigma_\src$ is strictly decreasing and convex so that \eqref{eq:effort_d}
has a unique solution for all $\bs{a}_s \in \mc{A}_s$. By definition, this solution is $\mu_s(\bs{a}_s)$. 
Implicit differentiation of~\eqref{eq:effort_d} then yields:
\begin{equation*}
\textstyle\frac{d
\mu_\src}{d\bs{\dpara}_\src} = \Big(2(\bs{\dpara}_\src)^2\Big(
\big(\frac{d}{d\mu_\src}\sigma_\src(\mu_\src) \big)^2 +
\sigma_\src(\mu_\src)\frac{d^2}{d^2
\mu_\src}\sigma_\src(\mu_\src)\Big)\Big)^{-1}
\end{equation*}
where we suppress the dependence of $\mu_\src$ on $\bs{\dpara}_\src$. 
The right-hand side of the above equation is strictly positive by
Assumption~\ref{ass:sigma_form}. Continuity follows directly by Assumption~\ref{ass:sigma_form}. 
\end{proof}

%% file: reformulation.tex
% !TEX root = root.tex

Finally, using our previous analysis and assumptions, we reformulate the optimization problem faced by each {\buyer}. This reformulation will simplify our analysis of equilibrium behavior in the data market, and lend economic interpretability to the results presented in Section \ref{sec:datam_results}.

Previously, we assumed that {\buyer} $b$'s estimator is separable in Assumption~\ref{ass:estimator_sep}. This allows us to write the loss function of $b$ as:
\begin{align*}
&L^\buy ( (c^\buy,\dpara^\buy), (c^{-\buy},\dpara^{-\buy}) )= 
\textstyle\sum_{i \in \Sc} h_\buy(x_i,x, F_\buy)
\sigma_i^2(\mu_i(\bs{\dpara}_i))\notag\\
&\quad\textstyle-\sum_{j \in -\buy} \comp_j^\buy \sum_{i \in \Sc} h_j
(x_i,x, F_j)\sigma_i^2(\mu_i(\bs{\dpara}_i))\notag\\
&\qquad\textstyle
  + \eta^\buy \sum_{i \in \Sc} \big( c_i^\buy  - \dpara_i^\buy \big[
      \sigma_i^2(\mu_i(\bs{\dpara}_i)) \notag\\
      &\quad\qquad\textstyle+ \sum_{l \in -i} h_b(x_l,
 x_{-i}, \delta_{x_i}) \sigma_l^2(\mu_l(\bs{\dpara}_l)) \big] \big)
\end{align*}
Recall that $x = (x_i)_{i \in \Sc}$ is fixed and common knowledge. Thus, we can replace each
of the evaluations of the $h_j$'s with constants. Towards this end, for each $i, l \in \mc{S}$ and $b \in \mc{B}$, we define:
\begin{equation}
\label{eq:betaparam}
\beta_i^b =h_b(x_i,x,F_b)
\end{equation}
\begin{equation}
  \xi_{i,l}^b=  \left\{\begin{array}{ll}
          h_b(x_l, x_{-i}, \delta_{x_i}) & \ i\neq l\\
          1 &\ i = l\end{array}\right.
\label{eq:betaxi}
\end{equation}
Note that each $\xi_{i,l}^j\geq 0$, by definition of the $\{h_b\}$. In addition, for each $i\in \mb{S}$ and
$b\in \mc{B}$, define:
\begin{equation}
\gamma_i^\buy =\textstyle \beta_i^\buy - \sum_{j \in-\buy} \comp_j^\buy \beta_i^j
\label{eq:gamma}
\end{equation}
Since we defined $\xi$ such that $\xi_{i,i}^\buy = 1$, we can write:
\begin{align*}
&L^\buy((c^\buy,\dpara^\buy), (c^{-\buy},\dpara^{-\buy}) ) =
\textstyle\sum_{i \in \Sc} \gamma_i^\buy \sigma_i^2(\mu_i(\bs{\dpara}_i))\notag\\
&\quad\textstyle  +\eta^\buy \sum_{i \in \Sc}
\left( c_i^\buy
 - \dpara_i^\buy \left[ \sum_{l \in
\Sc}  \xi_{i,l}^\buy \sigma_l^2(\mu_l(\bs{\dpara}_l))  \right] \right)
\label{eq:L}
\end{align*}
Similarly, the expected payment for any data source $\src$ and data
{\buyer} $\buy$  is given by:
\begin{equation*}
\textstyle
\mathrm{p}_\src^b( (c_\src^b,a_s^b), e)=
c_\src^\buy - \dpara_\src^\buy \left( \sum_{i \in \Sc} \xi_{\src,i}^\buy
\sigma_i^2(e_i) \right)
\end{equation*}
Before proceeding, we provide an interpretation of the constants introduced
above. The constant $\beta_s^b$ denotes the relevance of data sampled from
the point $x_s$  when constructing {\buyer} $b$'s estimate, given the distribution of all of the data sources. The parameter $\gamma_s^b$ corresponds to the level of demand 
that {\buyer} $b\in \mc{B}$ has for high-quality data from
source $\src \in \mc{S}$, factoring in the benefit this data supplies to
the competitors of $b$. In other words, $\gamma$ parameters capture the effects of the non-rivalrous nature of data. The parameter $\xi^b_{s,l}$ denotes a measure of coupling that exists between the payment contracts
 $\mathrm{p}^b_s$ and $\mathrm{p}^b_l$. In the case of a single {\buyer} (i.e.~\cite{cai:2015aa}),
this coupling did not prove problematic. In contrast, when there are multiple {\buyers}, each {\buyer} has an incentive to try and exploit this coupling, as shall become clear in our ensuing analysis. This coupling will play a central role in determining the existence and efficiency of equilibrium behavior in the data market. 

\begin{comment}
The parameter $\gamma_s^b$ corresponds to the level of demand 
that {\buyer} $b\in \mc{B}$ has for high-quality data from
source $\src \in \mc{S}$, factoring in the benefit this data supplies to
the competitors of $b$. In other words, $\gamma$ parameters capture the effects of the non-rivalrous nature of data. The parameter $\xi^b_{s,l}$ denotes a measure of coupling that exists between the payment contracts
 $\mathrm{p}^b_s$ and $\mathrm{p}^b_l$. In the case of a single {\buyer} (i.e.~\cite{cai:2015aa}),
this coupling does not exist as there is no competition between {\buyers}, and therefore does not lead to any inefficiencies. In contrast, when there are multiple {\buyers}, each {\buyer} has an incentive to try and exploit this coupling, giving rise to 
inefficiencies.
\end{comment}

Collecting the various expressions we have introduced,  {\buyer} $\buy$'s optimization problem can be re-written as: 
\begin{align}
\min_{(c^\buy,a^\buy)}
& \ L^\buy( (c^\buy,\dpara^\buy), (c^{-\buy},\dpara^{-\buy}) ) \label{eq:opt_p} \\
\text{s.t.} \sum_{j \in \Bc}& \left[c_\src^j - \dpara_\src^j \left( \sum_{i \in \Sc} \xi_{\src,i}^j
\sigma_i^2(\mu_i(\bs{\dpara}_i)) \right)\right]\geq \mu_\src(\bs{\dpara}_\src), ~ \forall s \in \mc{S}
\notag
\\
& c_\src^\buy - \dpara_\src^\buy \left( \sum_{i \in \Sc} \xi_{\src,i}^\buy
\sigma_i^2(\mu_i(\bs{\dpara}_i)) \right)\geq 0, ~ \forall s \in \mc{S}
\notag
\\
& \bs{\dpara}_\src \in \mc{A}_\src, ~ \forall s \in \mc{S} \quad \quad \dpara_\src^\buy \geq 0, ~ \forall s \in \mc{S} \notag
\end{align}
Without loss of generality, we let $\eta^\buy = 1$, by normalizing the $\gamma_\src^\buy$ accordingly. Note that the constraint $c_\src^\buy \geq 0$ can be omitted, in light of the constraint $c_\src^\buy - \dpara_\src^\buy \left( \sum_{i \in \Sc} \xi_{\src,i}^\buy\sigma_i^2(e_i^\ast) \right) \geq 0$, since each $\xi_{\src, i}^\buy \geq 0$ and $\dpara_\src^\buy \geq 0$.

%% file: notationtable.tex
\begin{table*}[h]
    \centering
    \begin{tabular}{|c|l|c|}
        \hline
        \textbf{Notation} & \textbf{Meaning} & \textbf{Defined or First Used in
        Equation} \\\hline
        $\src$ & index of data source & -- \\
        $\mc{S}$ & index set of data sources & --\\
        $\buy$ & index of {\buyer}& --\\
        $\mc{B}$ & index set of {\buyers} & --\\
      
        $\mathrm{p}_\src^\buy$ & expected payment from {\buyer} $\buy$ to source $\src$& \eqref{eq:exppay}\\
        $a_s^b$  & linear term in  $p_s^b$; used to adjust level of effort $e_s$ in equilibrium& \eqref{eq:paystructure} \\
        $a_s$ & vector containing the $a$ parameters offered to source $s$ by the members of $\mc{B}$ & --\\
        $c_s^b$ & constant term in $p_s^b$; used to ensure incentive compatibility in equilibrium & \eqref{eq:paystructure} \\
         $c_s$ & vector containing the $c$ parameters offered to sources $s$ by the members of $\mc{B}$ & --\\
        $\bs{a}_s$ & sum of $a$ parameters offered to source $s$ across all members of $\mc{B}$ & \eqref{eq:bara}\\  
        $\underline{\bs{a}}_s$ &  mimimum value of $\bs{a}_s$ required to ensure source $s$ does not exert negative effort & \eqref{eq:dmin}\\
        $\overline{\bs{a}}_s$ & minimum value of $\bs{a}_s$ at which data source $s$ exerts her maximum effort & \eqref{eq:dmax}\\
        $\mathcal{A}_s$ & $\left[ \underline{\bs{a}}_s,  \overline{\bs{a}}_s \right]$, the allowable range of $\bs{a}_s$ & -- \\
        $\mu_s$ & implicit map which returns the equilibrium value of $e_s$ as a function of $\bs{a}_s$ & -- \\ 
        $\zeta_j^b$ & level of competition between $j,\buy\in
        \mc{B}$ & \eqref{eq:buyer_obj}\\
        $\beta_s^b$ & relevance of data from $x_s$ in constructing {\buyer} $b$'s estimator & \eqref{eq:betaparam}\\
        $\gamma_\src^\buy$ & aggregate demand for $e_\src$ from $\buy$&\eqref{eq:gamma} \\
        $\bs{\gamma}_s$ & sum of demand for data source $s$ across all members of $\mc{B}$& \eqref{eq:gamma_tot}\\
        $\xi^b_{s,l}$ & coupling between $p^b_s$ and $p^b_l$& \eqref{eq:betaxi}\\
        \hline
    \end{tabular}
    \caption{Notation Reference Chart}
    \label{tab:notation}
\end{table*}

%% file: gnep.tex
% !TEX root = root.tex

It is important to note that the constraints each {\buyer} faces in her optimization problem~\eqref{eq:opt_p} depend on the actions taken by the rest of the {\buyers} in the data market. In particular, in order to ensure that the IR and IC constraints are maintained in equilibrium, we require an equilibrium concept which allows each {\buyer}'s admissible action space to depend on the choice of contract parameters selected by the other {\buyers} in the data market. Thus, we will employ the notion of a \emph{generalized Nash equilibrium} \cite{facchinei2007generalized} to study competitive outcomes in the data market, which is a natural extension of the typical notion of Nash equilibrium to this setting. 

Let $\Z_b\subset \mb{R}^{2N}$ be {\buyer} $b$'s actions space; that is, $z^b=(
c^b, a^b)\in \mb{R}^{2N}$ where $c^b=(c^b_s)_{s\in \mc{S}}$ and
$a^b=(a^b_s)_{s\in \mc{S}}$.
Each {\buyer} $b\in \mc{B}$  solves a parametric nonlinear programming problem
given by
\begin{equation}
P_b(z^{-b}) := \min_{z^b}~\{L^b(z^b,z^{-b}) : z^b\in \mc{M}^b(z^{-b})\}
\label{eq:gnep1}
\end{equation}
where 
   $\mc{M}^b(z^{-b}) = \{z^b : \ k_j^b(z^b,z^{-b})\geq 0,\ \ \forall \ j\in
   \mc{J}^b\}\subset \Z_b$
with $\mathcal{J}^b=\{1, \ldots, |\mathcal{J}^b|\}$ a finite set
indexing  the constraint functions of {\buyer} $b$.
Note that, unlike in the classic definition of a Nash equilibrium, 
 the admissible action space of {\buyer} $b$ depends on $z^{-b}$, the actions of $-b=\mathcal{B}
 \setminus
 \{b\}$.
 
We say $\{P_b( \cdot )\}_{b\in \mc{B}}$ is a generalized Nash (GN)
equilibrium problem.
A GN equilibrium is defined as follows.
 \begin{definition}
   A point $z=(z^{1}, \ldots, z^{|\mc{B}|}) \in \prod_{b = 1}^{|\mc{B}|}\Z_b$ is
   said to be a GN equilibrium for
$\{P_b( \cdot )\}_{b\in \mc{B}}$ if
%is defined to be 
for all $b\in \mc{B}$, $z^{b}$ solves $P_b(z^{-b})$.
     \label{def:gne}
 \end{definition}

We now analyze the game between the {\buyers} utilizing the notion of a GN
problem and GN equilibrium. We will characterize the existence and uniqueness of GN equilibria in two scenarios. In Section~\ref{subsec:unbounded}, we will consider the case where the effort spaces of data sources are unbounded, i.e.~$\mc{E}_s = \R_{\geq 0}$. In Section~\ref{subsec:bounded}, we will characterize the case where each data source has an upper bound on the level of effort they can exert, i.e.~$\mc{E}_s = [0,e_s^{max})$. In Section \ref{subsec:poa}, we will then address the social efficiency of the equilibria identified in Section \ref{subsec:unbounded}. A similar analysis of the equilibria identified in Section \ref{subsec:bounded} can be found in Appendix~\ref{sec:poa_bounded_case}.

Before preceding to out main results, we provide a technical lemma that will have a central role in our ensuing analysis and introduce some notation which will simplify the statement of our results. For compactness, for a given set of $a$ parameters we define $\mu(a) =(\mu_s(\bs{a}))_{s \in \mc{S}}$. (Recall that $\bs{a}$ is the sum of $a$ parameters, as defined in Equation~\eqref{eq:bara}.)
\begin{comment}
\begin{lemma}\label{lemma:continuous_obj}
    The objective function for each {\buyer} $\buy \in \mathcal{B}$, as given in
\eqref{eq:L}, is continuous with respect $z^\buy=(c^\buy,\dpara^\buy)$.
\end{lemma}
\end{comment}
\begin{lemma}
\label{lem:datam_eq_holds}
Suppose $z=(z^\buy)_{\buy\in \mc{B}}$, where $z^b = (c^b,a^b)$,  is a GN equilibrium for the game $\{P_b(\cdot)\}_{b\in \mc{B}}$ defined
by~\eqref{eq:opt_p}. Then for each $s \in \mc{S}$:
\begin{equation}
\textstyle    \sum_{j \in \mc{B}} c_\src^j
 -\sum_{j \in \Bc}\dpara_\src^j \big( \sum_{i \in \Sc}
\xi_{\src,i}^j \sigma_i^2(\mu_i(\bs{\dpara}_i)) \big)
-\mu_\src(\bs{\dpara}_\src)=0
\label{eq:bind}
\end{equation}
In other words, the IR constraint is always binding in equilibrium, and the expected payment to data source $s$ is equal to the effort exerted in equilibrium:
$\mathrm{p}_\src((c_\src,\dpara_\src),\mu(a)) = \mu_s(\bs{a}_s)$
\end{lemma}
\begin{proof}
Suppose that there is an equilibrium
in which~the IR constraint is not binding for some data source
$\src$. Then, there must a exists an {\buyer} $\buy$ whose non-negativity constriant corresponding to source $\src$ is
also not binding. Thus, this cannot be an equilibrium as {\buyer} $\buy$ can
unilaterally improve their payoff by decreasing $c_\src^\buy$ without causing
any of the constraints to be violated, contradicting the assertion
that the given selection of parameters is an equilibrium.
\end{proof}

The result of Lemma~\ref{lem:datam_eq_holds} is a well-known result in 
contract design---that is, the individual rationality constraint always
binds for the optimal contract~\cite{bolton:2005aa}. As shall become clear in our analysis in the following sections, the equality \eqref{eq:bind} forms an implicit constraint that appears in each of the {\buyers}' optimizations, which will be directly responsible for the degeneracy observed in the data market. Roughly speaking, while the $a$ parameters selected by the {\buyers} determine the level of effort that the data sources will exert, the $c$ parameters determine what portion of this effort each {\buyer} is expected to compensate.

For each $\src\in \mc{S}$, define:
\begin{equation}\label{eq:gamma_tot}
  \textstyle  \bs{\gamma}_\src = \sum_{j \in \mathcal{B}}\gamma_\src^j
\end{equation}
which can be interpreted to be the total demand for high quality data from data source $s$.
Next, we define:
\begin{equation*}
\bs{c}_s = \sum_{j \in \mc{B}} c_s^j
\qquad 
q_\src^\buy(\dpara)=  \dpara_s^\buy 
\left( 
\sum_{i \in \Sc}
\xi_{\src,i}^j \sigma_i^2(\mu_i(\bs{\dpara}_i))
\right)
\end{equation*}
\begin{equation*}
q_\src(\dpara) = 
\sum_{j \in \Bc} q_\src^j(\dpara) 
+\mu_\src(\bs{\dpara}_\src)
\end{equation*}
Note that Lemma~\ref{lem:datam_eq_holds} implies that if $(c,a)$ is an GN equilibrium in the game between the {\buyers} then $\bs{c}_s = q_s(a)$ will hold for each $s \in \mc{S}$. Moreover, the non-negativity constraints in the game between the buyers will hold only if $c_s^b \geq q_s^b(a)$ for each $s \in \mc{S}$ and $b \in \mc{B}$.

%% file: unbounded.tex
% !TEX root = root.tex

Let us first consider the case where there is no upper bound on the effort the data
sources may exert, i.e.~$\mc{E}_\src=\mb{R}_{\geq 0}$. 
\begin{theorem}
\label{thm:nash}
Consider the game $\{P_b( \cdot )\}_{b\in \mc{B}}$, where each {\buyer}'s objective is to solve the optimization
in~\eqref{eq:opt_p}. Suppose that for each $\src\in\mc{S}$,
$\mathcal{E}_\src = \R_{\geq 0}$ and $\bs{\gamma}_\src \geq
\underline{\bs{\dpara}}_\src$. Further, suppose that 
 $\gamma_i^j > 0$,  $\forall \ i \in \Sc, j \in \Bc$. 
%Further assume that $\forall \src\in \mc{S}$,
%$$.
Then, there is either no GN equilibrium or an infinite number of
GN equilibria. Moreover, if $(\bar{c},\bar{a})$ is a GN equilibrium, then the following conditions hold:  
\begin{enumerate}
\item The set of infinite GN equilibria is given by:
\begin{align*}
\textstyle    \big\{(c,a) : a =\bar{a},\ 
\bs{c}_s = q_\src(\bar{a}),\ c_s^b  \geq q_{\src}^\buy(\bar{\dpara}),  \ \forall s, \forall b \}
\label{eq:polytope}
\end{align*}
That is, the $a$ parameters selected by the {\buyers} are the same across each GN equilibrium, and all degeneracy lies in the equilibrium $c$ parameters which lie in the $|B|$-dimensional convex polytope defined above.
\item The effort exerted by each data source is the same in each GN equilibrium and the efforts constitute a unique induced dominated strategy equilibrium
between the data sources. More precisely, each data source exerts effort $\mu_\src(\bar{\bs{\dpara}}_\src)$ in all GN equilibria.
\end{enumerate}
\end{theorem}

Before going ahead with the proof of the theorem, we discuss its hypotheses and implications. The hypothesis that $\bs{\gamma}_\src \geq
\underline{\bs{\dpara}}_\src$ implies that there is enough demand for the data from source $s$ such that she does not have incentive to exert negative effort in equilibrium. Together, the {\buyers} will provide sufficient incentive to $s$ so that $s$ accepts each of the contracts offered to her, and truthfully report her query-response. When we investigate the case where $s$ only provides readings to a subset of the {\buyers} in Section \ref{subsec:partial}, only the relevant subset of {\buyers} must maintain this constraint. This condition places a restriction on what subsets of incentives from the {\buyers} each data source is willing to accept. 

As we discovered in Section \ref{subsec:sourcegame}, the $a=(a^{\buy})_{\buy\in \mc{B}}$ parameters selected by the {\buyers} uniquely determine how much 
effort the data sources exert in equilibrium. Intuitively, the fact that the $a$ parameters are constant across all GN equilibria means that, when GN equilibria do exist in the game between the {\buyers}, the {\buyers} have agreed to incentivize the data sources to each exert a particular level of effort. The proof of the theorem will shed some light on how this unique choice of $a$ parameters is selected when GN equilibria exist, and also demonstrate what `goes wrong' in cases where the {\buyers} cannot agree on how much effort to incentivize the sources to exert. In the latter case, no GN equilibrium solution exists in the game between the {\buyers}. Further commentary on this point is provided after the proof of Theorem~\ref{thm:nash_constrained}.

Meanwhile, for a fixed profile of $a$ parameters, the $c=(c^{\buy})_{\buy\in \mc{B}}$ parameters determine how much
of this effort each {\buyer} is responsible for compensating in expectation.
Even when {\buyers} are able to agree on how much effort to incentivize from the data sources and select the unique GN equilibria choice for $a$, 
there is a non-uniqueness in the $c$ parameters in equilibrium. This implies that there is a fundamental ambiguity in who will fund the exertion of the data sources. 
 In the extreme case, it is possible for one {\buyer} to pay for the entirety of
 the expected compensation offered to the data sources, while the other
 {\buyers} pay nothing in expectation. 

\begin{proof}[{Proof of Theorem~\ref{thm:nash}}]
By Lemma~\ref{lem:datam_eq_holds}, we have that:
\begin{equation}
\textstyle c_\src^\buy = \sum_{j \in \Bc}\dpara_\src^j \left( \sum_{i \in \Sc}
\xi_{\src,i}^j \sigma_i^2(\mu_i(\bs{\dpara}_i) ) \right) +\mu_\src(\bs{\dpara}_\src) -
\sum_{j \in -\buy} c_\src^j
\label{eq:cbuy}
\end{equation}
Plugging in this constraint, the cost function for {\buyer} $\buy$ can be expressed as:
\begin{align*}
\textstyle \tilde{L}^\buy &
\textstyle 
(\dpara^\buy, (c^{-\buy},\dpara^{-\buy}) ) 
=
\sum_{i \in \Sc}
\Big(
\gamma_i^\buy \sigma_i^2(\mu_i(\bs{\dpara}_i))
\\ &
\textstyle 
+ \sum_{j\in-\buy}
\left[
\dpara_i^j \left(
\sum_{l\in\mc{S}} \xi_{i,l}^j \sigma_l^2 (\mu_l(\bs{\dpara}_l)) 
\right) 
-c_i^j
\right]
+ \mu_i(\bs{\dpara}_i)
\Big)
\end{align*}
By swapping the roles of $i$ and $l$ in the middle term above, {\buyer} $\buy$'s cost
can be decomposed into the sum of costs for each data sources. We define:
\begin{align*}
    \tilde{L}^\buy_\src(\dpara^\buy_\src, (c^{-\buy},\dpara^{-\buy})
    )&\textstyle=\big(\gamma_\src^b+\sum_{j\in -\buy}\sum_{l\in
        \mc{S}}\dpara_l^j\xi_{l,s}^j\big)
        \sigma_\src^2(\mu_\src(\bs{\dpara}_\src))\\
        &\textstyle\quad-\sum_{j\in
        -\buy}c_\src^j+\mu_\src(\bs{\dpara}_\src)
\end{align*}
Then {\buyer} $\buy$'s optimization problem reduces to:
\begin{align*}
    \min_{\dpara^\buy} \ & \textstyle\sum_{\src\in \mc{S}}\tilde{L}_\src^\buy(\dpara_\src^\buy,
    (c^{-\buy},\dpara^{-\buy}) )\\
    \text{s.t.}\ &\  \textstyle\sum_{j\in
-\buy} 
\big[
\dpara_\src^j \left( \sum_{i \in \Sc}
\xi_{\src,i}^j \sigma_i^2(\mu_i(\bs{\dpara}_i) ) \right)-c_\src^j
\big] +\mu_s(\bs{\dpara}_\src)\geq
0\\
&\ \bs{\dpara}_\src\in \mc{A}_\src, \ \dpara_\src^\buy\geq 0
\label{eq:newopt2}
\end{align*}
Note that the cost does not depend on $c_s^b$, for any $s$. 
We complete the argument by ignoring the constraints and showing that the constraints are
satisfied for the set of equilibria we characterize.

%Note that the cost function for {\buyer} $\buy$ does not depend on
%    $c_\src^\buy$, for any $\src$, and 
%    thus 
%    this we can safely ignore the constraint $c_i^k -
%    d_i^k \left( \sum_{l \in \Sc}  \xi_{i,l}^k \sigma_l^2(e_l^*) \right) \geq
%    0$, since we can find values of $d^j$ that minimizes the objective, and then
%    find values for $c^j$ that satisfy this constraint. We also temporarily
%    ignore the constraint $d_q^{total} \in D_q$, and shall reconsider it
%    shortly.Thus, in order to find optimality conditions we simply differentiate
%the cost with respect to $d_q^k$:}
%\textcolor{red}{In this argument we are ignoring the constraints on $\dpara$,
%    i.e. $\bs{\dpara}_\src\in \mc{A}_\src$, $\dpara_\src^\buy\geq 0$}.

Differentiating the cost with respect to $\dpara_{\src}^\buy$ and
applying~\eqref{eq:effort_d} and $\xi_{\src,\src}^j = 1$ for all $j$, we
have that:
    \begin{align*}
        %\pd{\tilde{L}_{\src}^\buy}{\dpara_{\src}^\buy}
      \textstyle  D_{\dpara_{\src}^\buy}\tilde{L}_{\src}^\buy&=
      \textstyle\frac{1}{\bs{\dpara}_\src} \left(
            \dpara_\src^\buy-\gamma_\src^\buy -\sum_{j \in -\buy}\sum_{l \in
            -\src}\dpara_l^j \xi_{l,\src}^j
        \right)D_{\bs{\dpara}_\src}\mu_\src
        %\frac{d\mu_\src}{d\dpara_\src^\buy}
    \end{align*}
    where $D_x\equiv \frac{\partial}{\partial x}$.
%\[
%\frac{\partial}{\partial d_q^k} \sum_{i \in \Sc} 
%J_i^k(d_i^k,c^{-k},d^{-k}) =
%\frac{\partial}{\partial d_q^k} J_q^k(d_q^k,c^{-k},d^{-k}) =
%\]
%\[
%\left( \gamma_q^k +  \sum_{j \in -k}\sum_{l \in \Sc}d_l^j \xi_{l,q}^j \sigma_q(e_q^*)\right)2\sigma_q(e_q^*)\frac{d}{de_q^*}\sigma(e_q^*)\frac{de_q^*}{dd_q^k}+  \frac{de_q^*}{dd_q^k}=
%\]
%\[
%\left(1 - \frac{ \gamma_q^k + \sum_{j \in -k}\sum_{l \in \Sc}d_l^j \xi_{l,q}^j }{d_q^{total}} \right)\frac{de_q^*}{dd_q^k}=
%\]
%\[
%\left( \frac{ d_q^{total}-\gamma_q^k - \sum_{j \in -k}\sum_{l \in \Sc}d_l^j \xi_{l,q}^j }{d_q^{total}} \right)\frac{de_q^*}{dd_q^k}=
%\]
%\[
%\left(\frac{ d_q^k-\gamma_q^k -\sum_{j \in -k}\sum_{l \in -q}d_l^j \xi_{l,q}^j }{d_q^{total}}\right)\frac{de_q^*}{dd_q^k}
%\]
%In the third equality we us equation~\eqref{eq:effort_d}, and in the fifth equality we use the fact that $\xi_{q,q}^j = 1$ for all $j$. 
Applying Lemma~\ref{lemma:increasing_effort}, which states
$D_{\bs{\dpara_\src}}\mu_\src>0$, we get the following conditions:

\begin{equation*}
%    \label{eq:unbdd_cases}
\begin{cases}
    D_{\dpara_{\src}^\buy}\tilde{L}_{\src}^\buy < 0, & \text{if } 0 \leq \dpara_\src^{\buy} < \gamma_\src^{\buy} + \sum_{j \in
-\buy}\sum_{l \in -\src}\dpara_l^j \xi_{l,\src}^j \\
D_{\dpara_{\src}^\buy}\tilde{L}_{\src}^\buy= 0, & \text{if } \dpara_\src^{\buy} = \gamma_\src^{\buy} + \sum_{j \in
-\buy}\sum_{l \in -\src}\dpara_l^j \xi_{l,\src}^j \\
D_{\dpara_{\src}^\buy}\tilde{L}_{\src}^\buy> 0, & \text{if } \dpara_\src^{\buy} > \gamma_\src^{\buy}+ \sum_{j \in
-\buy}\sum_{l \in -\src}\dpara_l^j \xi_{l,\src}^j 
\end{cases}
\end{equation*}

Hence, the $\dpara_\src^{\buy}$ that minimizes {\buyer} $\buy$'s cost satisfies:
\begin{equation}
\label{eq:optimal_d}
\textstyle\dpara_\src^\buy = \gamma_\src^\buy + \sum_{j \in -\buy}\sum_{l \in
-\src}\dpara_l^j \xi_{l,\src}^j
\end{equation}
Performing this analysis for all combinations of $\src \in \mathcal{S}$ and
$\buy \in \mathcal{B}$ yields a system of $M\times N$ equations  with $M \times N$ unknowns, of the form~\eqref{eq:optimal_d}.

%\tc{red}{This is confusing, for the column vector what is being indexed? I mean
%why is there $i$ and $j$?}
 Let ${\dpara}$ denote a column vector with entries
$\dpara_i^j$ and let ${\gamma}$ denote a column vector containing all the terms of the form $\gamma_i^j$ for each $i \in \Sc$ and $j \in \Bc$. Then,
\eqref{eq:optimal_d} can be written as the system of
equations given by:
\begin{equation}
\label{eq:eq_conditions}
\dpara = \Amat \dpara + {\gamma}
\end{equation} 
Here, $\Xi$ is a non-negative matrix whose entries are composed of the
$\xi_{i,l}^j$ values such that~\eqref{eq:eq_conditions} expresses the set of
equality constraints defined by~\eqref{eq:optimal_d} for all $\src \in \Sc$ and
$\buy \in \Bc$. 
To solve this reduced game, it suffices to find a
solution to~\eqref{eq:eq_conditions} such that $\bs{\dpara}_i\in \mc{A}_i$ and $\dpara_i^j \geq 0$ for all $i\in
\mc{S}$ and $j\in \mc{B}$.

%\textcolor{red}{L: the following doesnt guarantee that the solution will be such
%    that $\dpara_i^j\geq 0$ and $\bs{\dpara}_i\in \mc{A}_i$.}
Let us consider first solutions to the system of equations
\eqref{eq:eq_conditions}.    
Systems of equations of this form are well studied in the economics literature,
as they are of the form specified by the celebrated Leontief input-output model.
It has been shown that such systems of equations have a non-negative solution if
and only if $\rho(\Amat) < 1$, where $\rho(\Amat)$ is the spectral radius of
$\Amat$~\cite{Stanczak2006}. Moreover, if such a solution exists, it must be unique.

Thus, if $\rho(\Amat) < 1$, inversion of the $(I -\Amat)$ matrix yields the
equilibrium $\dpara$, and, by Lemma~\ref{lem:datam_eq_holds}, we can pick any
$c$ such that for each $s\in \mc{S}$, $\buy\in \mc{B}$,
\begin{equation}
\textstyle    c_\src^\buy \geq \dpara_\src^\buy \left(\sum_{l \in \Sc}
    \xi_{\src,l}^\buy\sigma_l^2(\mu_l(\bs{\dpara}_l))\right)
   % c_\src^\buy \geq \dpara_\src^\buy \left( \sum_{l \in \Sc} \xi_{\src,l}^\buy
   % \sigma_l^2(\mu_l(\bs{\dpara}_l))
   % \label{eq:ineq}
\end{equation}
and \eqref{eq:cbuy} hold. 

%\[
%\sum_{j \in -\Bc} c_i^j = \sum_{j \in \Bc}d_i^j \left( \sum_{l \in \Sc}  \xi_{i,l}^j \sigma_l^2(e_l^*) \right) +e_l^*
%\]
%\[
%c_i^k \geq d_i^k \left( \sum_{l \in \Sc} \xi_{i,l}^k \sigma_l^2(e_l^*) \right).
%\]

If $\rho(\Amat) \geq 1$, there will not exist a non-negative solution and there
is no point $({c},{\dpara})$ that simultaneously optimizes~\eqref{eq:opt_p} for
all $\buy\in \mc{B}$. Finally, we demonstrate that none of the solutions to
\eqref{eq:eq_conditions} and corresponding $c$ values defined above
%equilibria we have pinpointed thus far  
violate the constraint $\bs{\dpara}_\src \in \mc{A}_\src =
[\underline{\bs{\dpara}}_\src, \infty )$. By inspecting equations of the
    form~\eqref{eq:optimal_d}, we see that $\bs{\dpara}_\src \geq
    \bs{\gamma}_\src \geq \underline{\bs{\dpara}}_\src$, and thus the
    constraints remain satisfied. It follows that there is either a unique set
    of
    ${\dpara}$ parameters defining a GN
    equilibria for the reduced game between the {\buyers}, otherwise there is no GN equilibrium. Moreover, in the case that an equilibrium choice of $\dpara$ does exist, by inspection we see that the polytope of $c$ parameters defined in the statement of the theorem constitute GN equilibria for the full game. 
    %
    %that will constitute a Nash equilibrium for the game along with a convex polytope of potential GNE, or there this no solution to the game, as desired.
\end{proof}
%\tc{red}{I dont see where we are invoking Theorem 1 in this proof above to even
%ensure there exists a Nash.}
It is interesting to note that the existence of GN equilibria
depends solely on the value of the $\xi_{\src,l}^\buy$ parameters; it does not
depend on the magnitude of the $\gamma_\src^\buy$ parameters (given that they are large enough to ensure participation of all parties). This implies that
the existence of GN equilibria 
%is simply an
%artifact of 
follows from the form of the contract mechanisms, and does not depend on whether or not there are solutions that are beneficial to all parties involved. 

%% file: bounded.tex
% !TEX root = root.tex

Let us now consider the case where the data sources' effort space is upper-bounded, i.e.~$\mc{E}_s=[0,e_\src^{\max}]$, $0\leq e_\src^{\max}<\infty$.

\begin{theorem}\label{thm:nash_constrained}
Consider the game  $\{P_b(\cdot)\}_{b\in \mc{B}}$ where each {\buyer}'s objective is to solve the optimization
in~\eqref{eq:opt_p}. 
 Suppose that for each $\src\in\mc{S}$,
 $\mathcal{E}_\src = [0,e_\src^{\max}]$ with $0\leq e_\src^{\max}<\infty$
 and that $\gamma_i^j > 0$ for all $i \in \Sc, j \in \Bc$. 
There is  an infinite number of
GN
equilibria $z=(z^\buy)_{\buy\in \mc{B}}$.
Moreover, the following statements hold:
\begin{enumerate}[topsep=0pt, itemsep=0pt]
\item There may exist two GN equilibria $(c_1,a_1)$ and $(c_2, a_2)$ such that $a_1 \neq a_2$; 
\item If $(\bar{c}, \bar{a})$ constitutes a GN equilibrium then the following set of parameters also constitute GN equilibria:
\begin{align*}
\big\{(c, a) : a = \bar{a},  \ \forall \src\in \mc{S}, \buy\in \mc{B}, \
c_s = q_\src(\bar{a}), \ c_s^b \geq q_s^b(\bar{a})\big\}
\end{align*}
\item While the data sources may exert different levels of effort across different
    equilibria, each collection of ${\dpara}$ parameters chosen by the {\buyers} still
    induce a dominanted strategy equilibrium between the data sources. 
\end{enumerate}
\end{theorem}

\begin{proof}
    Following the proof of Theorem \ref{thm:nash}, the problems $P_b(\cdot)$
    can be reduced to 
the optimization problem 
\begin{equation*}
    \textstyle \tilde{P}_\buy(z^{-b}):= \min_{\dpara^b}~\big\{ 
        \tilde{L}^b(\dpara^b,z^{-b}):\forall s\in \mc{S}, \bs{\dpara}_s \in \mc{A}_s, \
\dpara^b
\geq 0 \big\}
\end{equation*}
To show existence, we show that the game defined by
$\{\tilde{P}_b(\cdot)\}_{\buy\in \mc{B}}$
satisfies the assumptions
of Theorem~\ref{thm:gne_existence}, which is originally from \cite{Park:2015aa} and can be found in the Appendix.

First, we note that the objective function of each buyer is continuous in each of its arguments. Indeed,  Assumption \ref{ass:sigma_form} ensures that $\sigma_s^2$ is continuous and Lemma \ref{lemma:increasing_effort} ensures the map $\mu_s$ is continuous for each $s \in \mc{S}$. Continuity of the objective function follows by recalling that the composition of continuous maps yields a continuos map. 
Next, for each $j\in \mc{B}$, the constraints defining 
$\mathcal{M}_{-j}(z^{-j})$ are continuous, and thus the correspondence $z^{-j}
\mapsto \mathcal{M}_{-j}(z^{-j})$ is upper semi-continuous (indeed, even
continuous). Continuing with the analysis from Theorem \ref{thm:nash}, it is
easy to see that in the case where we constrain $\bs{\dpara}_i \in \mc{A}_i$,
the best response set for {\buyer} $b$ is defined by the following conditional
statements for each $\src\in \mc{S}$: 
\begin{equation}\label{eq:br_const}
\begin{cases}
    \dpara_\src^b =\underline{\bs{\dpara}}_\src -\sum_{j \in -b} \dpara_s^b & \text{if} \
T_\src^b(\dpara^{-b}) < \underline{\bs{\dpara}}_\src \\ 
\dpara_\src^b =\gamma_\src^b + \sum_{j \in -b}\sum_{l \in -\src}\dpara_l^j
\xi_{l,\src}^j &
\text {if} \ T_{\src}^b(\dpara^{-b}) \in \mc{A}_\src\\
\dpara_\src^b =\overline{\bs{\dpara}}_\src -\sum_{j \in -b} \dpara_s^j & \text{if} \
T_\src^b(\dpara^{-b}) > \overline{\bs{\dpara}}_\src
\end{cases}
\end{equation}
where $T_\src^b(\dpara^{-b}) = \gamma_\src^b + \sum_{j \in -b}\sum_{l \in
-\src}\dpara_l^j \xi_{l,\src}^j + \sum_{j \in -b}\dpara_s^j$.

Thus, the best response set for {\buyer} $b$ is always a singleton which in turn
implies it is always
contractable. Further, each best response mapping $\text{BR}_b
\colon \dpara^{-b} \mapsto \dpara^b$ is continuous in $\dpara^{-b}$. 
Hence,
\begin{equation}
    \min_{\dpara^b} L^b(\dpara^b, \dpara^{-b} ) = L^b(\text{BR}_b(\dpara^{-b}) ,
    \dpara^{-b} )
\end{equation}
and thus the mapping $\dpara^{-b }\mapsto \min_{\dpara^b} L^b(\dpara^b,
\dpara^{-b} )$ is continuous since $L^b$ and $\text{BR}_b$ are continuous and the
composition of continuous maps is continuous. Thus, by Theorem~\ref{thm:gne_existence} there exists a GN equilibrium for the reduced game between the {\buyers}. As in the proof of
Theorem~\ref{thm:nash}, we see that for any collection of
$\dpara$
parameters that constitute a GN equilibrium for this simplified
game, any collection of $c$ parameters that lie in the convex polytope defined in the statement of the theorem
constitute a GN equilibrium for the full game. 
\end{proof}

We now remark on why there are always GN equilibria to the game between the
{\buyers}
when each source can exert a finite amount of effort (i.e.~Theorem
\ref{thm:nash_constrained}),
but there may not be a GN equilibrium when the sources are allowed to exert
infinite effort (i.e.~Theorem \ref{thm:nash}).
Referring to~\eqref{eq:eq_conditions}, consider the case where $\rho(\Xi) = k <
1$ and GN equilibria exist. Suppose we replace $\Xi$ in~\eqref{eq:eq_conditions} with
$\alpha \Xi$, where $\alpha >1$, and note that $\rho(\alpha \Xi) = \alpha
\rho(\Xi)$. As $\alpha$ is increased towards $\frac{1}{k}$,
the matrix ${(I - \alpha \Xi)}$ gets closer to becoming singular, and the
corresponding solution to the system of equations, $a$, approaches infinity.
This implies 
 the data sources exert an infinite amount of effort in equilibrium. 

Intuitively, this corresponds to the coupling between the payment contracts
approaching some critical limit, past which point no GN equilibria  
exist. However, in the case where the data sources are constrained to exert a
finite amount of effort, this {run-away} behavior
is not possible, and GN equilibria always exist.  While the constraints bounding the effort the data
sources may exert ensure the existence of GN equilibria, their activation and inactivation may lead to a degeneracy in how much effort each data source
exerts in equilibrium. 

Note that, in Theorem \ref{thm:nash} and Theorem \ref{thm:nash_constrained}, we
assume that either all of the data sources are constrained in their effort,
or all unconstrained in their effort. In the case
where some data sources are constrained and others unconstrained, the task of
determining existence of equilibria to the game becomes a combinatorial
endeavor. We exclude the analysis of this case, as it is largely an algebraic
exercise and lends little insight to the broader problem.

%% file: poa.tex
% !TEX root = root.tex

The question of equilibrium \emph{efficiency} or \emph{quality} arises naturally
in game theoretic settings. 
In this section, we identify necessary and sufficient conditions under which the equilibria are socially inefficient; as we will see in Theorem~\ref{thm:always_lose}, as soon as any non-diagonal $\xi$ parameter is non-zero, there will be social inefficiency.

Due to the wide variety of estimators data {\buyers} can use, as well as effort-to-variance functions that characterize data sources, it is difficult to provide interesting general bounds on the price of anarchy, a widely used metric for the inefficiency of equilibria~\cite{roughgarden2007}. However, when both are specified, the price of anarchy can be explicitly calculated~\cite{westenbroek:2017aa}.

In this section, we will focus our attention on the case where the data sources are unconstrained in the effort they exert, i.e.~when $\mc{E}_s$ is unbounded, as the results in this case provide a clearer intuition for how our chosen class of mechanisms give rise to inefficiencies. However, in the interest of completeness, in Appendix~\ref{sec:poa_bounded_case} the result is extended to the case where data sources are effort-constrained. 

Let us denote by $\bs{e}$ the vector denoting the level of effort the data sources
exert. The social cost is defined as the sum of the cost experienced by all parties. 
 \begin{definition}[Ex-ante Social Cost]
\label{def:social_loss}
Suppose 
that  $\eta^j = 1$ for each {\buyer} $j\in \mc{B}$. We define the {ex-ante social cost} to be the sum of
the utility functions of all the data {\buyers} and data sources---that is:
\begin{align}
\textstyle\mathcal{L}(\bs{e}) &= \textstyle\sum_{j \in \mathcal{B}} \big(
\E \big[ 
\big( \hat{f}^j_{\mc{X}_j}(x^\ast) - f(x^\ast) \big)^2 -\notag\\
&\textstyle\quad\sum_{k \in -j} \comp_k^j \big( \hat{f}^k_{\mc{X}_k}(x^\ast) -
f(x^\ast) \big)^2
\big]
\big) + \sum_{s \in \mc{S}} e_s
\end{align}
\end{definition}

Note that this sum does not include any of the payments made in the marketplace, as they are simply lossless transfers of wealth. 
The fact that $\eta^j = 1$ for each $j\in
\mc{B}$  ensures that these transfers of wealth are lossless from a utility
perspective---i.e.~the {\buyers} and sources value the payment equally. However,
this is simply a rescaling and, more importantly, it allows us to isolate the social loss due to the mechanism, and ignore any losses due to differential preferences in payment currency. Indeed, note that the social cost only depends on the effort exerted by the sources.  

The following result states that there is always a unique level of effort that minimizes the social cost. 
\begin{lemma}\label{lemma:convex_social}
    Suppose that
$\gamma_i^j > 0$,  $\forall j \in \mathcal{B}$ and $\forall i\in \mc{S}$. There
is a unique minimizer of $\mc{L}(\bs{e})$.
\end{lemma}
\begin{proof}
The ex-ante social cost can be re-written as $\mathcal{L}\left(\bs{e}\right) = \sum_{i \in \mathcal{S} }\sum_{j \in \mathcal{B}} \gamma_i^j \sigma_i^2(e_i) + \sum_{i \in \mathcal{S}} e_i$. 
By Assumption~\ref{ass:sigma_form} and the assumption that $\gamma_\src^j > 0$
for some $j\in \mc{B}$, $D_{e_\src}^2\mc{L}(\bs{e})$ is strictly positive.
In addition,
$D_{e_\src} \mathcal{L}(\bs{e})$ does not depend on
$e_{i}$ for $i \in \mc{S}\setminus\{\src\}$.

Hence, $D^2_{\bs{e}}\mathcal{L}$ has positive entries on the diagonal and entries
    of zero everywhere else so that it is positive definite which, in turn,
    implies $\mathcal{L}$
    is strictly convex. Thus, since $\mathcal{E}$ is a convex set, $\mc{L}$ has a unique minimizer on
    $\mathcal{E}$. 
\end{proof}

The {price of
anarchy}  is defined
 for each Nash equilibrium as the ratio of the {social cost} under the Nash
 equilibrium to the socially optimal cost.
 \begin{definition}[Ex-ante Price of Anarchy]
     The {ex-ante price of anarchy} is given by $\PoA(\bs{e}) =\frac{\mathcal{L}(\bs{e})}{\mathcal{L}(\hat{\bs{e}})}$, 
where $\hat{\bs{e}} \in \mathcal{E}$ is the minimizer of $\mc{L}$.
\end{definition}

Since $\hat{\bs{e}}$ is unique minimizer of $\mathcal{L}$, for all
$\bs{e} \neq \hat{\bs{e}}$ we must have that $\PoA(\bs{e}) > 1$.
Intuitively, the larger $\PoA(\bs{e})$, the more socially inefficient the
solution is. With this metric in mind, we provide necessary and sufficient conditions
for the game between the {\buyers} to yield a socially efficient outcome, in the case that the effort space for each data source is unbounded. 

\begin{theorem}
\label{thm:always_lose}
Suppose $\gamma_i^j> 0$,  $\forall i \in \mc{S}$, $\forall j \in \mc{B}$. Further suppose that $\mc{E}_s = [0,\infty)$ for each $s \in \mc{S}$.
Then, there exists a GN equilibrium to the game $\{P^b{ (\cdot)} \}_{\buy\in
    \mc{B}}$ between the {\buyers} for which the price of anarchy is equal to one if and only if
for each $j \in \mc{B}$ and each $i,l\in \mc{S}$ such that $i \neq l$ we have that $\xi_{i,l}^j =0$.
\end{theorem}
\begin{proof}
    When there is no GN equilibria of the {\buyers}'
    game, the
proof is trivial. On the other hand, when there is GN equilibria
we have that
\begin{equation}\label{eq:social_effort}
\textstyle    D_{e_\src} \mathcal{L}(\bs{e})=
    2\bs{\gamma}_\src\sigma_\src(\hat{e}_\src)\frac{d}{de_\src}\sigma_\src(\hat{e}_\src) + 1 = 0.
\end{equation}
Since $\mc{L}$ is strictly convex, solving the first order conditions in
\eqref{eq:social_effort} yields the global minimizer. 
Just as with~\eqref{eq:effort_d}, the solution to~\eqref{eq:social_effort},
$\hat{e}_\src\in \mb{R}_+$, is
implicitly defined by $\hat{\mu}_\src:
\bs{\gamma}_\src\mapsto \hat{e}_\src$. 

Moreover, at a GN equilibrium,
we have $2\bs{a}_s \sigma_s(e_s)\frac{d}{de_s}\sigma_s(e_s)+1=0$.
Thus, as a consequence of
Lemma~\ref{lemma:increasing_effort}, the data sources  exert $\hat{e}_\src\in \mb{R}_+$ 
if and only if $\bs{\dpara}_i =
\bs{\gamma}_i$ for all $i \in \mathcal{S}$. 
By the proof of Theorem \ref{thm:nash} an equilibrium choice of the parameters $(c,a)$ must satisfy:
\begin{equation}
   \textstyle \bs{\dpara}_i= \bs{\gamma}_i +  \sum_{j \in \mathcal{B}} \sum_{k
   \in -j}\sum_{l \in
        -i}\dpara_l^k \xi_{l,i}^k
\end{equation}
Thus, we will have
$\bs{\dpara}_i =
\bs{\gamma}_i$ if and only if:
\begin{equation}
  \textstyle  0=\sum_{j \in \mathcal{B}} \sum_{k \in -j}\sum_{l \in
        -i}\dpara_l^k \xi_{l,i}^k
        \label{eq:zero}
\end{equation}
However, $\dpara_s^b> \gamma_s^b>0$ for some $s\in \mc{S}$ and $b \in \mathcal{B}$ 
since in equilibrium we have $a_\buy^s=\gamma_s^b+\sum_{j\in -b}\sum_{l\in -s}a_{l}^j\xi_{l,s}^j$
(see \eqref{eq:br_const}), and by assumption $\xi_{l,i}^j>0$ for some $j\in \mc{B}$ and some $i,l\in
\mc{S}$ with $i\neq l$. Thus, \eqref{eq:zero} cannot hold.
\end{proof}
Theorem~\ref{thm:always_lose} 
confirms the typical result that Nash
equilibria are generally (ex-ante) inefficient. However, it further shows that, in the
particular case of this framework, GN equilibria are efficient only when
there is no coupling between the payments the {\buyers} make to data sources. In this light, we find it appealing to regard the data sources as public good which the {\buyers} have incentive to exploit.  As the following result demonstrates, this problematic coupling will always arise when the aggregators utilize linear regression.

\begin{corollary}
    Suppose that each {\buyer} $b\in \mc{B}$ has estimator $\hat{f}^b$ which
    is linear regression. Further suppose that the conditions of Theorem \ref{thm:always_lose} hold.
    Then there does not exist a GNE solution to $\{P^b(\cdot)\}_{b \in \mc{B} }$  wherein the price of anarchy is equal to $1$. 
\end{corollary}
\begin{proof}
It can be shown (see \cite[Footnote 6]{cai:2015aa}) that 
\begin{align*}
g_b(x_{-s}, \delta_{x_s},
    &\sigma^2_{-s}(e_{-s}))=\mb{E}_{\tilde{x}\sim \delta_{x_s}}\big[
        [\tilde{x}^T,1]\cdot(X^TX)^{-1}X^T\notag\\
    &\cdot
    \diag(\sigma^2_{-s}(e_{-s}))\cdot X(X^TX)^{-1}\cdot [\tilde{x}^T,1]^T\big]
\end{align*} 
where $X$ is the matrix whose rows are $[x_i^T, 1]$ for $i\in -s$ and
$\diag(\sigma^2_{-s}(e_{-s}))$ 
is the diagonal matrix whose $(i, i)$--th entry is $\sigma_i^2(e_i)$.
By inspection, $g_\buy(x_{-\src}, \delta_{x_\src},
\sigma^2_{-\src}(e_{-\src}))$ is ill-defined if $[x_s^T,1]$ is orthogonal to
$\mathrm{span}\{([x_{i}^T,1])_{i\in -s}\}$. This, in turn, implies the payment
contract $p_s^b$ is ill-defined. Since all the payments are
    assumed well-defined, $[x_s^T,1]$ cannot be orthogonal to
$\mathrm{span}\{([x_{i}^T,1])_{i\in -s}\}$. Thus by inspection there exists $i, l \in \mc{S}$ such that $i \neq
    l$ and $\xi_{i,l}^b > 0$. 
Thus, by Theorem~\ref{thm:always_lose}, there is no efficient equilibrium.
\end{proof}

%% file: partial.tex
% !TEX root = root.tex

Thus far we have made the restrictive assumption that each of the data sources accepts payment from and provides query responses to each of the {\buyers} in the data market. We have done this primarily to ease the introduction of the notation needed to state and prove our main results. In this section, we remove this assumption and analyze the case where each of the data sources only accepts payment contracts from and provides query responses to a subset of the data {\buyers}. That is, we now assume that prior to the first stage outlined at the beginning of Section \ref{sec:datam_formulation} each data source $s \in \mc{S}$ has agreed to accept the incentives issued by some subset of the buyers $\mc{B}_s \subset \mc{B}$. As we shall see, these changes do not alter our previous analysis significantly. When {\buyer} $b \in \mc{B}$ only  purchases data from a subset of the data sources, the primary difference is that $b$ now has fewer contract parameters $(c^b,a^b)$ and fewer IR and non-negativity constraints in his optimization. The removal of these terms and constraints reduces the dimensionality of the degeneracy observed in the equilibria of the datamarket, but the overall structure of our analysis changes little. In order to demonstrate this point, in this section we will focus on demonstrating in some detail how a result analogous to Theorem \ref{thm:nash} can be obtained in this setting.

In practice, the mechanism by which data sources choose which incentives to accept and which {\buyers} to work with could be quite complicated. Although this is an important and interesting point for future research, in this paper we will assume that the sets $\mc{B}_s$ are exogenously given for all $s$. Our purpose here is to show that for any exogenously fixed assignments of data sources to data {\buyers}, 
the degeneracies we highlighted earlier remain whenever at least one data source receives payment from and provides data to multiple {\buyers}, i.e.~whenever the non-rivalrous nature of data has an effect on the data market.\footnote{An interesting question for future work is the study of how these assignments $\mc{B}_s$ would come to be in real-world settings, as well as the identification of socially desirable assignments. Our results here provide evidence that it will likely be difficult to find these desirable assignments.} Throughout the section we will outline how the formulation we have considered must be modified to fit this more general setting, and discuss how the results we have presented thus-far carry through. As we shall see, the generalization is rather straightforward, and thus some details are omitted in the interest of brevity.

Once it has been decided which data sources will provide data to which {\buyers}, the interactions of the data market proceed as before. Each {\buyer} $b \in \mc{B}$ issues incentives $p^b = (p_s^b)_{s \in \mc{S}_b}$ of the form \eqref{eq:paystructure} to the members of $\mc{S}_b$, then each data source $s \in \mc{S}$ evaluates the payments $p_s =(p_s^j)_{j \in \mc{B}_s}$, decides what level of effort to exert when producing $y_s$ and then shares this reading with the members of $\mc{B}_s$. Each {\buyer} then processes the data she has received to construct her estimate for $f_b$, issues payments $p^b$, and incurs loss $L^b$. 

Note that we have abused notation in redefining $p_s$ and $p^b$ above to only reflect the subset of payments that are issued in this section. Similar abuses will follow as  we redefine a number of objects from earlier in the document to be appropriate for this new setting. Roughly speaking, each of these items will be redefined by replacing $\mc{B}$ with $\mc{B}_s$ and $\mc{S}$ with $\mc{S}_b$ where appropriate. For example, the buyers now need only to select the parameters $c ^b = (c_s^b)_{ s \in \mc{S}_b}$ and $a^b = (a_s^b)_{s \in \mc{S}_b}$ when issuing incentives. We will omit the details of some of these changes when context makes our meaning clear.

We may now model the utility for each data source $s$ by: 
\begin{equation}
 u_s(e_s, p_s)=
    \E  \Big( \sum_{j \in \mc{B}_s} p_s^j( y^j(e)) \Big) - e_s
\end{equation}
and the loss for each {\buyer} $b$ by:
\begin{align*}
\textstyle L^\buy({p^\buy}, e)  &=\textstyle
\E \Big[ 
    \big( \hat{f}^\buy_{\mc{X}_\buy}(x^\ast) - f(x^\ast) \big)^2 \notag \\
    & \textstyle \quad -\sum_{j \in -\buy} \comp_j^\buy \big( \hat{f}^j_{\mc{X}_j}(x^\ast) - f(x^\ast) \big)^2 \\
    & \textstyle \quad + \eta^\buy \sum_{\src \in \mc{S}_b} p_\src^\buy(y^b(e)) \notag
\Big]
\end{align*}
where we now adopt the convention that $\mc{X}^b = (x_s,y_s)_{s \in \mc{S}_b}$ and $y^b(e) = (y_s(e_s))_{s \in \mc{S}_b}$. 

Letting $x^b = (x_i)_{i \in \mc{S}_b}$, the expected value of the payment $p_s^b$ can now be calculated as 
\begin{align*}
\mathrm{p}_\src^\buy& ( (c_\src^b, \dpara_\src^b), e ) 
= \E [ p_s^b(y^b(e)) ] \notag \\
& =
c_\src^\buy-\dpara_\src^\buy\big(
\sigma_\src^2(e_\src) 
+ g_\buy(x^b_{\mc{S}_b \setminus \{s\}}, \delta_{x_\src},
  (\sigma_i^{2}(e_i))_{i \in \mc{S}_b \setminus \{s\}} \big)
\end{align*}
and the expected total payment $s$ receives will be ${\mathrm{p}_\src((c_\src,\dpara_\src),e)=\textstyle\sum_{b\in
        \mc{B}_s}\mathrm{p}_\src^\buy((c_\src^b,\dpara_\src^b),e)}$. With these refactored definitions, the individual rationality, non-negativity and incentive compatibility constraints on the buyers' optimization problems are still given by equations \eqref{eq:IR_conditions}, \eqref{eq:IR_conditions2} and \eqref{eq:IC}, respectively. 

Now letting $\textstyle\bs{\dpara}_\src = \sum_{j \in \mathcal{B}_s} \dpara_\src^j$,
it is straightforward to show that the first-order optimality condition in \eqref{eq:effort_d} holds, and the ensuing analysis in Section \ref{subsec:sourcegame} pulls through. That is, one can show that each selection of parameters $a = (a^b)_{b \in \mc{B}}$ still induces a game between the data sources for which there is a dominant strategy equilibrium. Moreover, for each $s \in \mc{S}$, there exists an implicitly defined map $\mu_s \colon \R_{\geq 0} \to  \R_{\geq 0}$ which returns the equilibrium level of effort $e_s^*$ for each choice of $\textstyle\bs{\dpara}_\src $. The constants $\underline{\bs{a}}_s$, $\overline{\bs{a}}_s$ and the set $\mc{A}_s$ can also be redefined for this setting in a natural way. 

Following steps similar to those in Section  \ref{subsec:reformulation}, the loss for {\buyer} $b$ can now be written as:
\begin{align}
&L^\buy((c^\buy,\dpara^\buy), (c^{-\buy},\dpara^{-\buy}) ) =
\textstyle\sum_{i \in \Sc_b} \gamma_i^\buy \sigma_i^2(\mu_i(\bs{\dpara}_i)) \notag\\ 
&\quad + \sum_{i \in \mc{S}\setminus \Sc_b} \gamma_i^\buy \sigma_i^2(\mu_i(\bs{\dpara}_i)) \notag\\
&\quad\textstyle  +\eta^\buy \sum_{i \in \Sc_b}
\left( c_i^\buy
 - \dpara_i^\buy \left[ \sum_{l \in
\Sc_b}  \xi_{i,l}^\buy \sigma_l^2(\mu_l(\bs{\dpara}_l))  \right] \right)\label{eq:cost2}
\end{align}
where we define
${\beta_s^b = h_b(x_s,x^b,F_b)}$ and then define
${\gamma_s^\buy =\textstyle \beta_s^\buy - \sum_{j \in \mc{B}_s\setminus \{b\}} \comp_j^\buy \beta_s^j}$ for each $b \in \mc{B}$ and $s \in \mc{S}_b$.
Note that if $s \not \in \mc{S}_b$, then we do not need to define the constant $\beta_s^b$, since the query response that $s$ produces does not factor into the estimator that {\buyer} $b$ constructs for $f$. On the other hand, we do need to define $\gamma_s^b$ for each $s \in \mc{S}$ and $b \in \mc{B}$, since we have assumed that $s$ has agreed to provide at least one {\buyer} with the reading $y_s$. However, we note that the second term on the right hand side of \eqref{eq:cost2} does not depend on any of {\buyer} $b$'s decision variables.

Similarly, we only need to define the parameter $\xi_{i,l}^b$ if both $i,l \in \mc{S}_b$. In the case that data sources $i$ and $l$ both accept payment from {\buyer} $b$ we then define:
\begin{equation}
  \xi_{i,l}^b=  \left\{\begin{array}{ll}
          h_b(x_l, x_{-i}^b, \delta_{x_i}) & \ i\neq l \\ 
          1 &\ i = l
          \end{array}\right.          
\label{eq:betaxi}
\end{equation}

Applying the preceding analysis,  {\buyer} $\buy$'s optimization problem can now be re-written as: 
\begin{align}
\min_{(c^\buy,a^\buy)}
& \ L^\buy( (c^\buy,\dpara^\buy), (c^{-\buy},\dpara^{-\buy}) ) \notag \\
\text{s.t.} & \sum_{j \in \Bc_s} \left[c_\src^j - \dpara_\src^j \left( \sum_{i \in \Sc_j} \xi_{\src,i}^j
\sigma_i^2(\mu_i(\bs{\dpara}_i)) \right)\right]\geq \mu_\src(\bs{\dpara}_\src), \notag \\
& \qquad \qquad \qquad \qquad\qquad\qquad\qquad\qquad\qquad\forall s \in \mc{S}_b
\notag
\\
& c_\src^\buy - \dpara_\src^\buy \left( \sum_{i \in \Sc_b} \xi_{\src,i}^\buy
\sigma_i^2(\mu_i(\bs{\dpara}_i)) \right)\geq 0, ~ \forall s \in \mc{S}_b
\label{eq:opt_p2}
\\
& \bs{\dpara}_\src \in \mc{A}_\src, ~ \forall s \in \mc{S}_b \quad \quad \quad \dpara_\src^\buy \geq 0, ~ \forall s \in \mc{S}_b \notag
\end{align}
where $L^b$ is now defined by \eqref{eq:cost2}. Note that the optimization facing {\buyer} $b$ in \eqref{eq:opt_p2} is quite similar to the optimization in \eqref{eq:opt_p}, save the modifications to some IR and non-negativity constraints. As we shall see in the statement of Theorem \ref{thm:nash3} and Corollary \ref{corollary1} below, when there are GN equilibria in the game between the buyers, the removal of some of these constraints will affect the degeneracy previously seen in the $c$ parameters.

Before stating our generalization to Theorem \ref{thm:nash} we redefine some final notation. First, we define:  $\textstyle  \bs{\gamma}_\src = \sum_{j \in \mathcal{B}_s}\gamma_\src^j$, where we emphasize that $\textstyle  \bs{\gamma}_\src$ does not depend on $\gamma_s^b$ if $ b \notin \mc{B}_s$, since this term will fall out when characterizing the optimality condition for {\buyer} $b$, and not affect the existence of GN equilibria. We then define for each $s \in \mc{S}$ and $b \in \mc{B}$: 
$q_\src^\buy(\dpara)=  \dpara_s^\buy 
\left( 
\sum_{i \in \Sc_b}
\xi_{\src,i}^j \sigma_i^2(\mu_i(\bs{\dpara}_i))
\right)$ and 
$q_\src(\dpara) = 
\sum_{j \in \Bc_s} q_\src^j(\dpara) 
+\mu_\src(\bs{\dpara}_\src)$.

\begin{theorem}
\label{thm:nash3}
Consider the game  $\{P_b(\cdot)\}_{b\in \mc{B}}$ where each {\buyer}'s objective is to solve the optimization
in~\eqref{eq:opt_p2}. Suppose that for each $\src\in\mc{S}$,
$\mathcal{E}_\src = \R_{\geq 0}$, $\bs{\gamma}_\src \geq
\underline{\bs{\dpara}}_\src$. Further, suppose that 
 $\gamma_s^b > 0$,  for each $s$ and $b$ such that $s \in \mc{S}_b$. If there exists a GN equilibrium $(\bar{c},\bar{a})$ then the following conditions hold:
\begin{enumerate}
\item The set of GN equilibria in the game is given by
\begin{align*}
\textstyle    \big\{(c,a) : a =\bar{a},  \  
\bs{c}_s = q_\src(\bar{a}),\ c_s^b  \geq q_{\src}^\buy(\bar{\dpara}),  \ \forall s, \forall b \in \mc{B}_s\}
\end{align*}
That is, the $a$ parameters selected by the {\buyers} are the same across each GN equilibrium, and for each $s \in \mc{S}$ the equilibrium $c_s = (c_s^b)_{b \in \mc{B}_s}$ parameters lie in the $|\mc{B}_s|$-dimensional polytope defined above.
\item The effort exerted by each data source is the same in each GN equilibrium and the efforts constitute a unique induced dominated strategy equilibrium
between the data sources. More precisely, each data source exerts effort $\mu_\src(\bar{\bs{\dpara}}_\src)$ in each GN equilibrium.
\end{enumerate}
\end{theorem}
The proof is almost exactly the same as the proof of Theorem~\ref{thm:nash} and is omitted here. In particular, note that the only difference we need to consider is that the {\buyers} now have fewer decision variables and fewer constraints on these decision parameters. The removal of these components manifests itself in the dimensionality of the polytope of equilibrium parameters. 
\begin{corollary}\label{corollary1}
Consider the game  $\{P_b(\cdot)\}_{b\in \mc{B}}$ where each {\buyer}'s objective is to solve the optimization
in~\eqref{eq:opt_p2}, and suppose the assumptions of Theorem \ref{thm:nash3} hold. If there exists a GN equilibrium solution $(\bar{c},\bar{a})$ then the following two statements are true:
\begin{enumerate}
\item If $|\mc{B}_s| =1$ for all $s \in \mc{S}$ then $(\bar{c},\bar{a})$ is the only GN equilibrium.
\item If there exists $s \in \mc{S}$ such that $|\mc{B}_s | \geq 2$ then there are an infinite number of GN equilibria.
\end{enumerate}
\end{corollary}
We state the previous result to emphasize that when even a single data source accepts incentives from more than a single {\buyer} an infinity of GN equilibria arise in the data market (given that any GN equilibria exist). Returning to Theorem \ref{thm:nash3}, we see complete degeneracy in the equilibrium $c$ parameters offered to any data source who sells data to multiple buyers. 

An analogous generalization to Theorem \ref{thm:nash_constrained} is also straightforward to obtain for the more general case we consider in this setting, though we omit it in the interest of brevity. The analysis conducted in the proof of Theorem \ref{thm:always_lose} also follows through in a natural way. In particular, we still observe that the first order optimality conditions for the {\buyers} will coincide with the socially efficient choice of pricing parameters if and only if each of the non-diagonal $\xi$ parameters is zero. In particular, this means that if $|\mc{B}_s| = 1$ for each $s \in \mc{S}$ the data market will achieve a socially efficient outcome, with regards to the exogenous assignment of data sources we have assumed has already occurred.

%% file: conclusion.tex
% !TEX root = root.tex

We analyzed the strategic interactions between multiple data aggregators who share a pool of data sources. Previous work showed that a single data aggregator can find unique solutions that achieve socially efficient outcomes, but we demonstrate that the same mechanisms will break down as soon as a second data aggregator enters the market. In particular, we show that there are either no GN equilibria or infinitely many, and these solutions are frequently socially inefficient. This highlights the need for further research into mechanisms for data markets when there are multiple purchasers. In particular, there is a need for mechanisms that can simultaneously handle moral hazard \emph{and} the non-rivalrous nature of data.

%% file: existence_thm.tex
% !TEX root = root.tex
\subsection{GNE Existence Result}

 \begin{theorem}[Existence of GN equilibria~\cite{Park:2015aa}] \label{thm:gne_existence}
 Consider a GN equilibrium problem $\{P^b(\cdot)\}_{b\in \mc{B}}$.
 Suppose that for each $b \in \mc{B}$, the following hold: i) the correspondence $\mc{M}^b\colon \prod_{i\in-b} \Z_i
         \rightarrow \Z_b$ is
     upper semi-continuous. ii) the map $L^b$ is continuous on the graph of $\mc{M}^b$.
iii) the map $z^{-b} \mapsto \min_{z^b \in \mc{M}^b(z^{-b})} L^b(z^{b},
     z^{-b})$ is continuous.
iv) for all $z^{-b}$, the best response set, $ \text{BR}_{b}(z^{-b}) = \arg\min  \{L^b(z^b,z^{-b}) : z^b\in
 \mc{M}^b(z^{-b})\}$
 is contractable.
 
 Then, there exists a GN equilibrium.
 \end{theorem}
 
 \subsection{Price of Anarchy with Bounded Effort Spaces}
 \label{sec:poa_bounded_case}
 \begin{theorem}
 \label{thm:always_lose2}
Suppose $\gamma_i^j> 0$,  $\forall i \in \mc{S}$, $\forall j \in \mc{B}$. Further suppose that $\mc{E}_s = [0,e_s^{max}]$ and that $\bs{\gamma}_s \in [\underline{\bs{a}}_s,\overline{\bs{a}}_s)$ for each $s \in \mc{S}$.
Then, there exists a GN equilibrium to the game $\{P^b{ (\cdot)} \}_{\buy\in
    \mc{B}}$ between the {\buyers} for which the price of anarchy is equal to one if and only if
for each $j \in \mc{B}$ and each $i,l\in \mc{S}$ such that $i \neq l$ we have that $\xi_{i,l}^j =0$.
\end{theorem}
\begin{proof}
In the case where $\mc{A}_s$ is bounded $\forall s \in \mc{S}$, we
make an analogous argument as was made in Theorem \ref{thm:always_lose}. Suppose that
$\bs{\dpara}_s = \bs{\gamma}_s,  \forall s \in \mc{S}$ as is needed for the socially optimal solution. Then, by assumption, for each $s\in \mc{S}$, we have $\underline{\bs{\dpara}}_s < \bs{\dpara}_s < \overline{\bs{\dpara}}_s$,
and thus, it must be true that in equilibrium the $a$--parameters  are given
by an equation of the form
$a=\Xi a+\gamma$, since the best response for each
{\buyer} is given by $a_s^b=\gamma_s^b+\sum_{j\in -b}\sum_{l\in -s}a_l^j\xi_{l,s}^j$, 
which is the second option in \eqref{eq:br_const}. However, again it cannot be the case that $\mc{\gamma}_s =\bs{a}_s$ for each $s \in \mc{S}$ if $\xi_{i,l}^j \neq$ for some $j \in \mc{B}$ and $i,l \in \mc{S}$ such that $i \neq l$.
\end{proof}